\documentclass[12pt, a4paper,leqno]{amsart}
\usepackage{amsmath,amsthm,amscd,amssymb,amsfonts,amsbsy,mathtools}
\usepackage{latexsym}
\usepackage{txfonts}
\usepackage{exscale}
\usepackage{float}
\usepackage{graphicx,subfigure}

\usepackage[colorlinks,citecolor=red,pagebackref,hypertexnames=false]{hyperref}

\calclayout
\allowdisplaybreaks

\numberwithin{equation}{section}

\usepackage{pgf}
\usepackage{color}

\theoremstyle{plain}
\newtheorem{theorem}[equation]{Theorem}

\newtheorem{lem}[equation]{Lemma}

\theoremstyle{definition}

\theoremstyle{remark}
\newtheorem{remark}[equation]{Remark}
\newtheorem{rem}[equation]{Remark}

\newcommand{\R}{\mathbb{R}}

\newcommand{\cH}{\mathcal H}

\newcommand{\wt}{\widetilde}

\newcommand{\cP}{{\mathcal P}}

\newcommand{\1}{{\mathbf 1}}

\newcommand{\Z}{{\mathbb{Z}}}

\newcommand{\N}{\mathbb{N}}

\newcommand{\ipc}[2]{\left \langle #1 , \ #2 \right \rangle }

\begin{document}

\title[Localization landscape]{The exponential decay of eigenfunctions for tight binding Hamiltonians via landscape and dual landscape functions}

\author[W. Wang]{Wei Wang}
\email[]{wang9585@umn.edu}
\address{School of Mathematics, University of Minnesota, Minneapolis, MN, USA}

\author[S. Zhang]{Shiwen Zhang}
\email[]{zhan7294@umn.edu}
\address{School of Mathematics, University of Minnesota, Minneapolis, MN, USA}

\newcommand{\Addresses}{{
  \bigskip
  \vskip 0.08in \noindent --------------------------------------
\vskip 0.10in

  \footnotesize

\medskip

}}
\maketitle

\begin{abstract}
We consider the discrete Schr\"odinger operator $H=-\Delta+V$ on a cube $M\subset \Z^d$, with periodic or Dirichlet (simple) boundary conditions. We use a hidden {\em landscape} function $u$, defined as the solution of an inhomogeneous boundary problem with uniform right-hand side for $H$, to predict the location of the localized eigenfunctions of $H$. Explicit bounds on the exponential decay of Agmon type for low energy modes are obtained. This extends the recent work of Agmon type  of localization in \cite{ADFJM-CPDE} for $\R^d$ to a tight-binding Hamiltonian on $\Z^d$ lattice. Contrary to the continuous case, high energy modes are as localized as the low energy ones in discrete lattices. We show that exponential decay estimates of Agmon type also appear near the top of the spectrum, where the location of the localized eigenfunctions is predicted by a different landscape function. Our results are deterministic and are independent of the size of the cube. We also provide numerical experiments to confirm the  conditional results effectively, for some random potentials. 
 \end{abstract}

\tableofcontents

\section{Introduction and main results}\label{sec:intro}
Localization of eigenfunctions is one the most important topics in mathematics and condensed matter physics. The term {\it localization}, roughly speaking, refers to the phenomenon that the eigenfunctions of  an elliptic operator concentrate on a narrow region in space and are (exponentially) small outside the region. Take the one-electron model of condensed matter physics for example,  for which the spectral and transport properties of the material are described by a (one-particle) Schr\"odinger operator $H=-\Delta+V$. The Laplacian $\Delta$ describes the kinetic energy of a free particle, and the potential $V$ the presence of the external field. For example, the choice of a periodic function $V$ can be used to describe a perfect crystal.  The pioneering work of P.W. Anderson \cite{An} back to 1958, says that randomness causes the states to (exponentially) localize 
due to the disorder of the background media. In the past several decades, the localization in a disordered system has attracted a lot of interests, \cite{Ab}. There are tremendously many beautiful results, e.g. \cite{KS,FS,CKM,AM,ASFH,EA,DS} in the discrete setting, \cite{FK1,FK2,DSS,BK} in the continuum setting, which are far from a complete list.

In 2012, a new theory was proposed by Filoche and Mayboroda \cite{FM-PNAS} to study the 
location of localized state. They introduced the concept of the landscape, which is the solution $u$ to $Hu=1$ for a Schr\"odinger operator $H=\Delta+V$ on a finite domain with appropriate boundary conditions. The landscape function has remarkable power in studying the eigenvalue problems of the original Hamiltonian. Many beautiful results were obtained based on this simple landscape function, both in mathematics and  physics, \cite{St, LS,PRL-3,DFM}. Recently, in \cite{ADFJM-PRL,ADFJM-CPDE,ADFJM-SIAM}, the authors propose a new framework by viewing $1/u$ as an effective potential.  In their work, the barrier of the effective potential  revealed detailed location and shape of the eigenfunctions. There are both rigorous proof of the exponential decay estimates of Agmon type in \cite{ADFJM-CPDE}, as well as systematical numerical results for studying eigenvalue problems (localization subregions, eigenvalue counting functions) using the effective potentials in physics \cite{ADFJM-PRL} and in computational mathematics \cite{ADFJM-SIAM}. The latter two papers contain numerous beautiful conjectures. Only a few of them have been proved rigorously.   

The quantum state of a particle can be described by a wave function in $L^2(\R^d)$, where most of the previous work on landscape theory focuses on. In the tight-binding approximation, only nearest-neighbor hopping is allowed for electron-electron interaction, which leads to a one-particle Hamiltonian in the form of a discrete Schr\"odinger operator on $\ell^2(\Z^d)$, see e.g. \cite{AW,Kir}. In this work, we consider a tight-binding Schr\"odinger operator $H=-\Delta+V$ on a finite box $M$ in $\Z^d$ for any $d$, with periodic, or Dirichlet (simple) boundary conditions. Previously in \cite{LMF}, the authors showed that the exact same landscape theory in the continuum setting  \cite{FM-PNAS} can be extended to  a tight-binding Hamiltonian defined on a discrete $\Z^1$ lattice, with the Dirichlet boundary condition.  We show in this work that not only the discrete analogue of the landscape function $\{u_n\}_{n\in M}$ exists in this more general setting, but more importantly, the effective potential $1/u_n$ can be used to obtain exponential  decay estimates of Agmon type for the eigenfunctions of $H$. The latter extends the Agmon type of localization results in \cite{ADFJM-CPDE} to the discrete lattice.

  We will describe and prove our main results for the periodic lattice first. The Dirichlet boundary will be discussed in Section \ref{sec:diri} later.    Let us introduce the following notations in order to state our main results. More precise settings of tight-binding Hamiltonians and landscape theory  will be given in Section \ref{sec:pre} and \ref{sec:LL}.  Let $M= \Z^d/K\Z^d\cong \{\bar 1,\cdots,\bar K\}^d$ be an integer torus, where $K\in\N$ and  $\bar k$ is the congruence class, modulo $K$. For simplicity, we will omit the overhead bar of $\bar k$ frequently whenever it is clear.   Let $V=\{v_n\}_{n\in {M}}$ be a potential satisfying $0\le v_n \le  V_{\max}$ and not identically $0$, for some $ V_{\max}>0$ (the strength of the potential).
We consider a tight binding Hamiltonian $H $, acting on $\ell^2({M})$ defined as:
\begin{align}
    (H \varphi)_n&=-\sum_{|m-n|_1=1}\left(\varphi_m-\varphi_n\right)\,+ v_n \varphi_n,\ \ n\in {M} \label{eq:opH-intro}
\end{align}
where $|n|_1:=\sum_{i=1}^d|n_i|$ is the $1$-norm on $\Z^d/K\Z^d$. We may think of $\varphi$ either as a periodic sequence $\varphi_n$ or as a periodic function $\varphi(n)$ indexed by $\Z^d$.  
We are interested in the inhomogeneous equation, with uniform right-hand side, $(H  u)_n=1,n\in M$. The solution,  $u=\{u_n\}$ which is the so-called {\it landscape function}, is the discrete analogue of the landscape function introduced in \cite{FM-PNAS}, and will serve as a new tool to study localization of the eigenvalue problems of the original Schr\"odinger operator $H $. 

The first result of this paper is the exponential decay estimates of Agmon type for $H $. 
\begin{theorem}[Theorem \ref{thm:Agmon}]\label{thm:Agmon-intro}
Suppose $K\ge 3,K\in\N$, $H \varphi=\mu\varphi$ on $\ell^2({M})$. For any $\delta>0$, and $\mu\le  V_{\max}-\delta$, let $h_n$ be the Agmon weight, defined as in \eqref{eq:Agmon-dis}, associated to $\mu,\delta$ and $u_n$ on ${M}$. Then
\begin{align}\label{eq:Agmon-loc-intro}
    \sum_{h_n\ge 1}e^{c_1 h_n}\varphi_n^2\, \le\, \frac{C_0 }{\delta}\,   \sum_{ {n\in {M}} }\varphi_n^2
\end{align}
where constant $c_1>0$ only depends on the dimension $d$, and constant $C_0>0$ only depends on $d$ and $V_{\max}$. 
\end{theorem}
\begin{remark}
The dimensional constants $c_1\sim 1/\sqrt{d}$, $C_0\sim d(V_{\max}+d)$ explicitly, and are independent of $K$. The Agmon weight $h_n$ is defined explicitly through $\mu,\delta$ and $u_n$. All these will be specified in Theorem \ref{thm:Agmon}.
\end{remark}
\begin{rem}
We are most interested in the so-called ``thermodynamic
limit'': We restrict the system to a finite but large cube, and study quantities of interest in this finite system. Finally, we let the cube grow indefinitely and hope that the quantity under consideration has a limit. Though we did not really consider ``thermodynamic
limit'' as $K\to \infty$ in rigor, it is crucial that our main results are independent of the size $K$ of the cube.  
\end{rem}

The work \cite{LMF} for the discrete $\Z^1$ lattice was done before \cite{ADFJM-PRL,ADFJM-CPDE,ADFJM-SIAM}. Some of the new framework regarding the effective potential $1/u_n$ was not discussed in \cite{LMF}, especially for the decay estimate of Agmon type that we obtained. The first motivation of the current work is to  have a relatively complete picture of the landscape theory for tight-binding Hamiltonians on $\Z^d$ for any $d$, see e.g. Theorem \ref{thm:landscape}. Based on the discrete version of the landscape theory, we obtained the above exponential decay estimates. Some of the ground truth about the discrete landscape function  will be used  in our continuation work for the landscape law for the integrated density of states on $\Z^d$ (in preparation).  

It is well known that all the eigenvalues of the discrete operator $H$ in \eqref{eq:opH-intro} are contained in $[0,4d+V_{\max}]$, for any $K$.  Similar to the continuous case, Theorem \ref{thm:Agmon-intro} works effectively in the regime $[0,V_{\max}-\delta]$ near the bottom of the spectrum.  Contrary to the continuous case, two different types of localization and landscape occur on the discrete lattice. One is for low energy states near the bottom. The other one is for high energy states near the top of the spectrum, using a different, so-called, dual landscape. The dual landscape was first discovered and studied in \cite{LMF}. The authors there showed that, in the discrete $1$-d model, similar to the low energy state, the localization of the high energy modes can be governed by a different landscape function. In the current work, we study further this concept of the dual landscape, as well as the Agmon type of localization near the top of the spectrum, in more general situations. For $H$ in \eqref{eq:opH-intro} and any even integer $K\ge3$, there is a unique dual operator $\wt H$ acting on $\ell^2(M)$ as
\begin{align}
    (\wt H \varphi)_n&=-\sum_{|m-n|_1=1}\left(\varphi_m-\varphi_n\right)\,+ (V_{\max}-v_n) \varphi_n,\ \ n\in {M} \label{eq:opH-dual-intro}
\end{align}
Consider a dual landscape equation $(\wt H \wt u)_n=1$ on $\ell^2({M})$ ($\wt u$ will be specified later in Theorem \ref{thm:Agmon-dual}).  One can obtain the exponential decay of the high energy modes of $H $, as a direct consequence of Theorem \ref{thm:Agmon-intro} and the dual landscape $\wt u_n$ for $\wt H $. 
\begin{theorem}[Theorem \ref{thm:Agmon-dual}]\label{thm:Agmon-dual-intro}
Suppose $K\ge 3$ is an even integer, and $H\varphi=\mu\varphi$ on $\ell^2({M})$.  For any $\delta>0$, and  $\mu\ge 4d+\delta$, let $\wt h_n$ be the Agmon weight associated to $\mu,\delta$ and $\wt u_n$ on ${M}$ (see definition  in \eqref{eq:Agmon-dis-dual}). Then
\begin{align}\label{eq:Agmon-loc-top-intro}
    \sum_{\wt h_n\ge 1}e^{c_1 \wt h_n}\varphi_n^2\, \le\, \frac{C_0} {\delta}\,    \sum_{ {n\in {M}} }\varphi_n^2
\end{align}
where $c_1,C_0$ are the same constants as in Theorem \ref{thm:Agmon-intro}. 
\end{theorem}

Similar to the continuous case, the discrete effective potential $1/u_n$ is a deterministic tool. The localization results of Agmon type  through the effective potential can be applied to the random media, while the estimates only guarantee decay of the eigenfunction insofar as the effective weight functions $h_n$ grow. 
 We can not use the effective potential to prove Anderson localization on $\Z^d$ for a random potential. The landscape theory provides a new framework to detect the geometry  of disordered materials. It is a complement of the probabilistic methods for studying Anderson model and random media. In Section \ref{sec:num}, we will see numerical experiments on the behavior of the landscape function, and the effective weights, in random media. 
 Numerical results were obtained on predicting the localization sub-regions using  
  the level sets of $1/u_n$ and $1/\wt u_n$, for a tight-binding Anderson model.  The numerical part is a companion of our proof for $\Z^d$, as well as a parallel work of the computations for $\R^d$ in \cite{ADFJM-SIAM}, which will be discussed in detail in Section \ref{sec:num}.
  
  Our work on $\Z^d$ is also greatly inspired by the recent study in physics of many-body localization landscape \cite{BLG} by  Balasubramanian, Liao and  Galitski. In \cite{BLG}, the localization landscape for the single-particle Schr\"odinger operator was generalized to a wide class of interacting many-body Hamiltonians. \cite{BLG}  introduced the so-called many-body localization landscape on a discrete graph in the Fock space, and obtained bounds on the exponential decay of the many-body wave-functions in the Fock space. Our choice of Agmon metric on the $\Z^d$ lattice is similar to the construction in \cite{BLG} on the general Fock-state graph. The Fock-state model is essentially derived from a many-body spin chain model on the physical lattice. To the best of our knowledge, so far, the localization landscape in the Fock space is not able to show any decay estimates on the original spin chain. This is another motivation for us to study a tight-binding model on $\Z^d$, which is a special case of the general Fock-state model.  Our ultimate goal is to study the many-body spin chains in the physical space, using the localization landscape on $\Z^d$. We refer readers to e.g. \cite{ANSS} and references therein, for more detail about spin-chain models and their relation to tight-binding Schr\"odinger operators. 

The paper is organized as follows. In Section \ref{sec:pre} and \ref{sec:LL}, we state  preliminaries for  tight-binding Hamiltonians and the landscape theory for the  periodic boundary condition.  We prove the exponential decay estimates of  Agmon type near the bottom of the spectrum in Section \ref{sec:pfAgmon}. In Section \ref{sec:dual-LL}, we discuss the dual landscape and the exponential decay estimates for high energy modes. In Section \ref{sec:diri}, we describe how the main results can be generalized also to the Dirichlet (simple) boundary condition. In the last section, we discuss more numerical results of the landscape theory on discrete lattices.  \\

\noindent{\bf Acknowledgments.} 
We thank Professor Svitlana Mayboroda for suggesting to us the problem of studying localization landscape for tight-binding Hamiltonians. We appreciate
her support from the very beginning of this project and for many useful suggestions and
comments. We thank Professor Douglas Arnold for reading our early manuscript and many useful discussions. The paper is also partially motivated by the recent work \cite{BLG} of  Balasubramanian, Liao, and Galitski on Many-Body Localization Landscape. We would like to thank these authors for sharing us the early manuscript of their paper. We would also like to thank Li Chen for many useful discussions. 

Wei Wang is partially supported by the NSF grant DMS-1719694 and Simons Foundation grant 601937, DNA. Shiwen Zhang is supported in part by the NSF grants DMS-1344235, DMS-1839077,  and Simons Foundation grant 563916, SM. 


\section{Exponential decay of eigenfunctions of Agmon type}\label{sec:Agmon}
\subsection{Preliminaries}\label{sec:pre}
In the tight-binding model, the Hilbert space is taken as the sequence space 
$
    \ell^2(\Z^d)=\{\, (\phi_i)_{i\in\Z^d}\, |\, \sum_{i\in\Z^d}|\phi_i|^2<\infty\, \}
$
where we may think of $\phi=(\phi_i)_{i\in\Z^d}$ either as a function $\phi(i)$ on $\Z^d$ or as a sequence $\phi_n$ indexed by $i\in\Z^d$. It is convenient to equip  on $\Z^d$ the $1$-norm: $ |n|_1:=\sum_{i=1}^d|n_i|$, which reflects the graph structure of $\Z^d$. Two vertices $m=(m_1,\cdots, m_d), n=(n_1,\cdots, n_d)\in \Z^d$ are called nearest neighbors,  if $|m-n|_1=1$. We also say the nearest neighbors $m,n$ are connected by an edge of the discrete graph $\Z^d$. We denote by $e_i=(0,\cdots, 0,1, 0,\cdots, 0),i=1,\cdots, d$ the standard basis vector of $\Z^d$. The discrete (graph) Laplacian $\Delta$ on $\Z^d$ is defined as usual, acting on $\phi=\{\phi_n\}_{n\in \Z^d}$, 
\begin{align}\label{eq:Lap}
    (\Delta\phi)_n=\sum_{|m-n|_1=1}\left(\phi_m-\phi_n\right)=&\sum_{|m-n|_1=1}\phi_m\, -\, 2d\phi_n \\
    =&\sum_{1\le i\le d}\left(\phi_{n+e_i}+\phi_{n-e_i}\right)\, -\, 2d\phi_n ,\ \ n\in\Z^d .\nonumber 
\end{align}
For a real sequence $(v_n)_{n\in\Z^d}$ on $\Z^d$, the potential $V$ is a multiplication operator acting on $\phi\in\ell^2(\Z^d)$ as $(V\phi)_n=v_n\phi_n$. $H=-\Delta+V$ is usually called the discrete Schr\"odinger operator on $\Z^d$. If one takes $v_n=v_n(\omega)$  as independent, identically distributed random variables (in some probability space),  the random operator $H(\omega)=-\Delta+V(\omega)$ is usually referred as the Anderson model. We refer readers to the lecture notes \cite{Kir} for more details and a complete introduction to tight-binding Hamiltonians and the Anderson model. 

Next, we consider the discrete Schr\"odinger operator $H$ restricted on a finite cube (box) in $\Z^d$.  Let $\Lambda=\Z/K\Z\cong\{\bar 1,\bar 2,\cdots,\bar K\}$, where $\bar k$, is the congruence class in $\Z$, modulo $K$. We will omit the overhead bar and write $k=\bar k$ for simplicity whenever it is clear.  Let ${M}=\Lambda \times \Lambda \times \cdots \times \Lambda \cong \Z^d/K\Z^d$ be the $d$-tuple of $\Lambda$. We may abuse the notation and denote by $|\cdot|_1$ the induced 1-norm of $\Z^d$ on the congruence class $M$, where for example we consider two points $(1,n_2,\cdots,n_d)$ and $(K,n_2,\cdots,n_d)$ are closest neighbors and have distance 1 to each other in $M$.  From now on, we will concentrate on the finite dimensional subspace $\ell^2({M})$ of $\ell^2(\Z^d)$. For simplicity, we will denote by $\cH:= \ell^2({M})\cong \R^{dK}$. $\cH$ is equipped with the usual inner product on $\R^{dK}$, which we denote by $\ipc{\cdot}{\cdot}_{\cH}$.
It is easy to check that for $\phi \in \ell^2(M)$
\begin{align}\label{eq:bdry-peri}
     \phi_{n}=\phi_{n+Ke_i},\ n\in M,\ i=1\cdots,d 
\end{align}
which is equivalent to say that $\phi$ satisfies the periodic boundary condition. The restriction of $V$ on ${M}$ is $V_{{M}}=\{v_n\}_{n\in {M}}$. We omit the sub index ${M}$ and still denote it by $V$. We also assume that  $0\le v_n \le  V_{\max}$ and is not identically $0$, for some $ V_{\max}>0$. 
 
Our main interest will be the following  tight binding Hamiltonian $H =H _{{M}}
=-\Delta+ V$ acting on $\ell^2({M})$ as:
\begin{align}\label{eq:opH}
    (H \phi)_n=(-\Delta \phi)_n+(V\phi)_n=-\left(\sum_{|m-n|_1=1}\phi_m\right)\, +2d\phi_n+ v_n \phi_n,\ \ n\in {M} 
\end{align}

   For the self-adjoint operator $H  $, we say $(\mu,\varphi)$ is an eigenpair of  $H  $ if $H  \varphi=\mu \varphi$ for some $\mu\ge 0$ and $\varphi \in \cH$. It is easy to check that all eigenvalues of $H$ in \eqref{eq:opH} are contained in $[0,4d+V_{\max}]$ for any $K$.

In general, Schr\"odinger operators $H=-\Delta+V$ acting on a finite dimensional space are essentially matrices. Take $\Z^1$ for example, it is easy to check that under the  periodic boundary condition \eqref{eq:bdry-peri}, $H  $, $\Delta $ and $V$ can be identified as $K\times K$ matrices, acting on $\vec{\phi}\in \R^K$, for which we abuse the notation and still denote them by $\Delta $ and $V$:
\begin{align}\label{eq:LapPeri}
    -\Delta  =\begin{pmatrix}
	2 & -1 &0  & \cdots &  -1 \\
	-1 &   2 &\ddots   &  \vdots \\
	0 & \ddots& \ddots&\ddots  & 0 \\
	\vdots  &\ddots & \ddots&  2   &-1 \\
	-1  &\cdots & 0&-1 & 2
\end{pmatrix}_{K\times K}, \ \ V=\begin{pmatrix}
	  v_1 & 0 &0  & \cdots &  0 \\
	0 &   v_2 & 0 &\ddots   &  \vdots \\
	0 & \ddots& \ddots&\ddots  & 0 \\
	\vdots  &\ddots & \ddots&  v_{K-1}   &0 \\
	0  &\cdots & 0&0 &   v_K
\end{pmatrix}_{K\times K}
\end{align}

  When $v_i$ are taken as i.i.d. random variables, Schr\"odinger operators (matrices) $H =-\Delta   +V$ are the standard 1-d Anderson models, restricted on a finite box. See e.g. in Figure \ref{fig:BV} where $v_i$ are i.i.d. random variables with a Bernoulli distribution. 
\begin{figure}
    \centering
    \includegraphics[width=\textwidth]{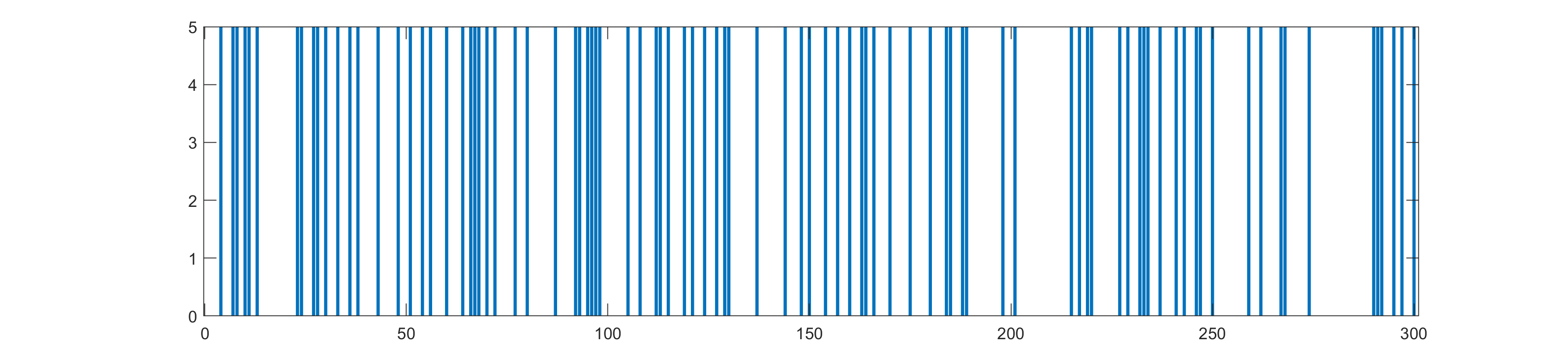}
    \caption{i.i.d. $\{0,5\}$ Bernoulli random potential with  probabilities $70\%$ and $30\%$, over a lattice of size $300$.}
    \label{fig:BV}
\end{figure}
Since the work of P.W. Anderson \cite{An}, many physicists have developed a fairly good knowledge
about the spectrum of random Schr\"odinger operators, though  only
 part of it has been shown with mathematical rigor. One of the most important phenomenon is called Anderson localization or exponential localization of the eigenfunctions of disordered systems. An example of the localization of the 1-d  random Schr\"odinger operator $H =-\Delta +V$ is given  in Figure  \ref{fig:phi14}.
\begin{figure}
    \centering
    \includegraphics[width=\textwidth]{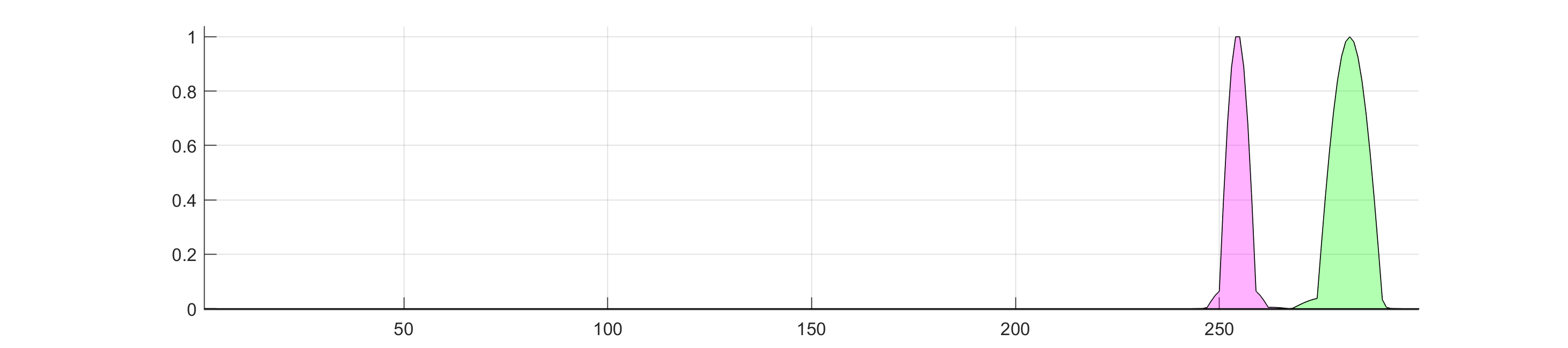}
    \caption{The first (green) and the fourth (pink)  eigenfunctions of the corresponding Schr\"odinger operator in \eqref{eq:LapPeri}, with $K=300$ and $v_n$ the Bernoulli potential in Figure \ref{fig:BV}. }
    \label{fig:phi14}
  \end{figure}	
  
  On the other hand, a fundamental result for the localization of eigenfunctions of the Schr\"odinger operators is Agmon's theory (see \cite{Ag,H}), which demonstrates an exponential decay for a large class of potentials. Roughly speaking, the exponential decay of eigenfunctions comes from the barrier of a general confining potential $V$.  One key step in such classical confinement is to study the so-called Agmon metric associated to the potential $V$ and energy $\mu$. Unfortunately, for the random potentials which we showed above, the Agmon metric is highly degenerate, and is not generally useful. 
  
Instead of looking at the original potential $V$, we are interested in the barrier and wells of the effective potential, $W_n=1/u_n,n=1,\cdots,K$, see Figure \ref{fig:BW}. The landscape function $u\in \R^K$ is the unique solution to the matrix equation $(-\Delta+V)u=\vec 1$, where $\vec 1=(1, \cdots,1)^T$ is the $K$-dimensional vector with all entries constantly one.
\begin{figure}
    \centering
    \includegraphics[width=\textwidth]{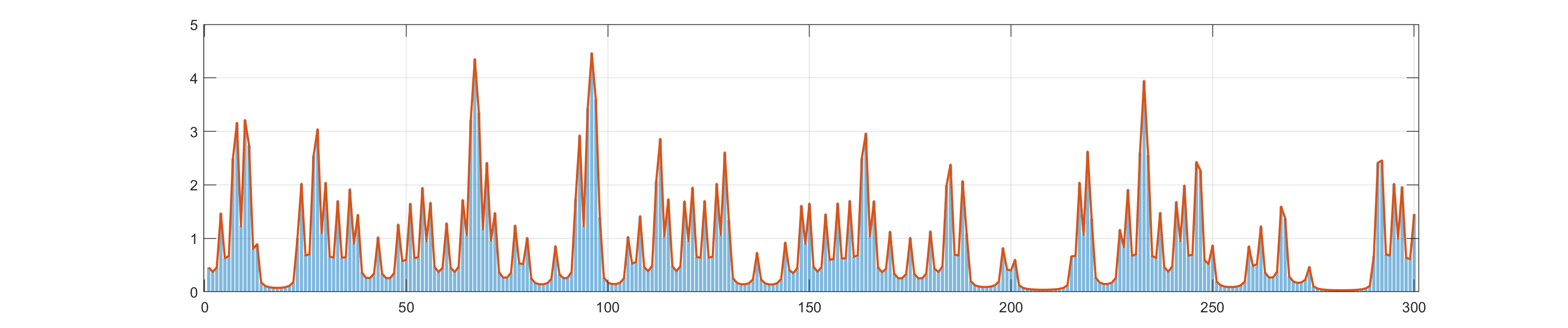}
    \caption{The effective potential $W_n=1/u_n$ associated to the Bernoulli potential of Figure \ref{fig:BV} }
    \label{fig:BW}
  \end{figure}

  Similar to the continuous case, the first result of this paper, Theorem \ref{thm:Agmon-intro} (see also Theorem \ref{thm:Agmon}), is the rigorous proof that the steep decay in
Figure \ref{fig:phi14} comes from the barriers of the effective potential $W_n$ in Figure \ref{fig:BW}. This is proved by formulating and proving appropriate exponential decay estimates of Agmon type.  We will see more numerical experiments and the comparison with our main results in Section \ref{sec:num}

\subsection{Landscape theory on discrete lattices}\label{sec:LL}

For a $\Z^d$ sequence $f=\{f_n\}_{n\in \Z^d}$,  its directional (right) derivatives and the gradient on $\Z^d$ are defined as \begin{align}\label{eq:nabla-i}
    \nabla_if_n=f_{n+e_i}-f_{n},\ \  i=1,\cdots, d, 
\end{align}
and its gradient $\nabla f:M \mapsto \R^d$ is
\begin{align*}
   \nabla f(n)=\left(\nabla_1f_n,\nabla_2f_n,\cdots,\nabla_df_n\right).
\end{align*}

We also denote the dot product and the induced norm of $\nabla f,\nabla g$, as $\R^d$ vectors, by
\begin{align*}
&(\nabla g \cdot \nabla f)(n)=\sum_{i=1}^d \nabla_i g_n\cdot \nabla_i f_n ,\ \ {\rm and} \ \  |\nabla f|(n):=  \sqrt{ (\nabla f \cdot \nabla f)(n)}
\end{align*}

\begin{lem}[Green's first identity]\label{lem:Dform}
Let $M=\Z^d/K\Z^d$. For any $f,g\in \cH=\ell^2({M})$, 
\begin{align}\label{eq:Dform}
    \ipc{g}{-\Delta f}_{\cH}  =\sum_{n\in {M}}(\nabla g \cdot \nabla f)(n) , \ \ \ 
    \ipc{f}{-\Delta f}_{\cH}=\ipc{|\nabla f|}{|\nabla f|}_\cH
\end{align}

\end{lem}

\begin{remark}
One can find the general discrete version of Green’s formula on graphs, for example in \cite{G,Chung}. We sketch the proof on $\Z^d$ lattices with the periodic boundary conditions for the reader's convenience. 
\end{remark}

\begin{proof}
For $n=(n_1,\cdots,n_d)\in\Z^d$, let $\check  n_i:=(n_1,\cdots,n_{i-1},n_{i+1},\cdots,n_d),\ i=1\cdots, d$. 
Direct computations show that for $i=1,\cdots,d$, 
\begin{align*}
    \sum_{n_i=1}^K g_n(- f_{n+e_i}-f_{n-e_i}+2 f_n)=&\sum_{n_i=1}^K g_n(- f_{n+e_i}+ f_n)+\sum_{n_i=1}^K g_n(-f_{n-e_i}+ f_n)\\
    =&-\sum_{n_i=1}^K g_n\nabla_if_n+\sum_{n_i=1}^K g_{n+e_i}(-f_{n}+ f_{n+e_i})\\
    =&-\sum_{n_i=1}^K g_n\nabla_if_n+\sum_{n_i=1}^K g_{n+e_i}\nabla_if_n\\
    =&\sum_{n_i=1}^K \nabla_ig_n\nabla_if_n
\end{align*}
where in the second sum, we used the periodicity to shift the sum in the i-th direction. 

Therefore, 
\begin{align*}
  \ipc{g}{-\Delta f}_{\cH}=  \sum_{n\in {M}}g_n\sum_{i=1}^d (- f_{n+e_i}-f_{n-e_i}+2 f_n)
    =&\sum_{i=1}^d  \sum_{\check n_i}\sum_{n_i=1}^K g_n(- f_{n+e_i}-f_{n-e_i}+2 f_n)\\
    =&\sum_{i=1}^d\sum_{\check n_i}\sum_{n_i=1}^K \nabla_ig_n\nabla_if_n          \\
    =&\sum_{i=1}^d\sum_{n\in M} \nabla_ig_n\nabla_if_n         =\sum_{n\in M} (\nabla g \cdot \nabla f )(n)
\end{align*}

\end{proof}

Notice that by \eqref{eq:Dform}, $\ipc{f}{\Delta f}=0$ implies $f_n=a$ is a constant vector for $n\in {M}$.  Therefore, $-\Delta $ is non-negative and the first (isolated) eigenvalue is $0$,  with eigenfunction $\vec 1\in \ell^2({M})$ such that  $\vec 1_n=1,\ \ n\in {M} $, i.e.,
\begin{align}\label{eq:vec1}
    \vec{1}\cong\underbrace{\begin{pmatrix}1\\\vdots\\1\end{pmatrix}_K\otimes \begin{pmatrix}1\\\vdots\\1\end{pmatrix}_K \otimes
    \cdots \otimes \begin{pmatrix}1\\\vdots\\1\end{pmatrix}_K}_{d}
\end{align}
Clearly, $\ipc{\vec{1}}{(-\Delta + V\vec{1})}_{\cH}= \ipc{\vec{1}}{  V\vec{1}}_{\cH}=  \sum_{n\in{M}} v_n\ge   \max v_n>0$, which also implies that  $\ipc{f}{(-\Delta + V)f}_{\cH}>0$ for all $0\neq f\in \ell^2({M} )$.   Therefore $H =-\Delta   +V$ is invertible on $\ell^2(M)$, which allows us to solve the inhomogeneous equation $H  u=\vec{1}$. We will refer the equation as the \emph{landscape equation} from now on. The unique solution $u=\{u_n\}_{n\in {M}}$,  will be called the {\emph{landscape function}} throughout the paper, which is the discrete analogue of the landscape function in the continuum case. We summarized it here as
\begin{theorem}\label{thm:landscape}
Assume that $v_n \ge 0$ and is not identically zero for $n\in {M}$. The inhomogeneous boundary value problem
\begin{align}\label{eq:landscape-eq}
 (H  \phi)_n=-(\Delta    \phi)_{n}+ v_n \phi_{n}=1,\ \ n\in {M}
\end{align}
has a unique solution $u=\{u_{n}\}\in \ell^2({M})$.  
\end{theorem}

The following maximum principle shows that the landscape function $u_n$ is positive everywhere in the cube ${M}$ with an explicit lower bound independent of the size of the cube. 

\begin{lem}[Maximum Principle]\label{lem:maxP}
Suppose $f\in \cH=\ell^2(M)$ satisfies that  $(H  f)_n\ge0,\  {n\in {M}} $, then $f_n\ge0$ for all $ {n\in {M}} $. 

As a consequence, if $(H  u)_n=1,n\in {M}$ for $u\in \cH$, then 
\begin{align}\label{eq:u-lower}
    \min_{ {n\in {M}} }u_n\ge 1/V_{\max}  >0.
\end{align}

\end{lem}

\begin{proof}
We shall prove the maximum principle first.  It is enough to consider $f\in \cH=\ell^2({M})$ not identically zero. If $f$ satisfies the periodic boundary condition \eqref{eq:bdry-peri}, it is more convenient to extend $f$ periodically to the entire $\Z^d$ lattice such that 
\begin{align*}
    f_{n+Ke_i}=f_{n}, \ \ i=1\cdots,d, \, n\in \Z^d
\end{align*}
The operator $H  $  is then extended to $\ell^2(\Z^d)$ automatically as in \eqref{eq:opH}, satisfying $(H  f)_n\ge 0,n\in\Z^d$.

Suppose $-a:=\min_{n\in\Z^d}f_n=\min_{n\in{M}}f_n<0$.  Let $j\in {M}$ be such that $f_{j}=-a<0.$ $(H   f)_j\ge 0$ implies that 
\begin{align*}
    2df_{j}+  v_j f_j\ge \sum_{1\le i\le d} (f_{j+e_i}+f_{j-e_i})\ge 2d f_j
\end{align*}

Notice that $v_j\ge 0$, $f_j<0$ and $f_{j\pm e_i}\ge f_j$, the only chance for the above inequality to hold is that $v_j=0$ and $f_{j+e_i}=f_{j-e_i}=-a$ for all $1\le i\le d$. Repeat the above argument for $f_{j+e_i},1\le i \le d$ since now $f_{j+e_i}=f_j<0$ is also a (global) minimum, we have $f_{j+2e_i}=f_{j+e_i}=f_j<0$. Inductively, one has  $v_n=0$ for all $n\in {M}$ and $f_n=f_j=-a<0$ for all $n\in {M}$, which is a contradiction. Therefore, $f_n\ge 0$ for all $n\in {M}$.

To obtain \eqref{eq:u-lower}, it is enough to pick an appropriate test function. Now suppose $u\in \cH$ satisfies $(H  u)_n= 1, n\in{M}$. Let $\vec 1$ be as in \eqref{eq:vec1}, and let $f_n=u_n-V^{-1}_{\max}\cdot \vec 1$. It is easy to check that 
\begin{align*}
    (H  f)_n=&(H  u)_n+V^{-1}_{\max}   \cdot(\Delta    \vec 1)_n-V^{-1}_{\max}   \cdot (V\vec 1)_n=1-V^{-1}_{\max}   \cdot {v_n}\ge 0
\end{align*}
since  $(\Delta\vec 1)_n=0$ for all $n\in {M}$

 By the maximum principle, $f_n=u_n-V^{-1}_{\max}  \ge 0$ for all $ {n\in {M}} $, which gives \eqref{eq:u-lower}. 
\end{proof}

The next lemma is a discrete analogue of Lemma. 4.1 in \cite{ADFJM-CPDE}, which is a new uncertainty principle regarding the effective potential  $W_n=1/{u_n}$.

\begin{lem}\label{lem:eff}
Let $H =-\Delta+V$ be as in \eqref{eq:opH}.  Suppose $(H  u)_n=1,n\in {M}$ for $u\in \cH=\ell^2(M)$.  For any $f,g\in \cH $,  one has
\begin{align}
    \ipc{g}{H  f}_\cH
         =&\sum_{n\in {M}}\,\sum_{1\le i\le d} \left(u_{n+e_i} u_n\cdot  \nabla_i\frac{ g_{n}}{u_n}\cdot  \nabla_i \frac{ f_{n}}{u_n} \right)\,+\, \sum_{n\in {M}}\, \frac{1}{u_n} \, g_{n}f_{n} \label{eq:eff}
\end{align}
where $\nabla_i\frac{ g_{n}}{u_n}=\frac{ g_{n+e_i}}{u_{n+e_i}}-\frac{ g_{n}}{u_n}$ is as in \eqref{eq:nabla-i}. 

In particular, 
\begin{align}
\ipc{f}{H  f}_\cH 
=\sum_{n\in {M}}\,\sum_{1\le i\le d} u_{n+e_i} u_n\cdot  \left(\nabla_i\frac{ f_{n}}{u_n}\right)^2
\, +\,\sum_{n\in {M}}\, \frac{1}{u_n}f_n^2 
\,\ge\, \sum_{n\in {M}}\frac{1}{u_n}f_n^2 \label{eq:eff1}
\end{align}
\end{lem}

\begin{proof}
The landscape equation \eqref{eq:landscape-eq} gives that for all $n\in {M}$, 
\begin{align}\label{eq:temp99}
    \frac{\sum_i u_{n+e_i}+u_{n-e_i}}{u_n}+\frac{1}{u_n}=2d+v_n .
\end{align}

For $f,g\in \cH  $,  we have that 
\begin{align*}
    \ipc{g}{H  f}_\cH  =&\ipc{g}{Vf}_\cH+\ipc{g}{-\Delta    f}_\cH\\
 =&\sum_{n\in {M}}v_n g_{n}f_{n}\,+\, \sum_{n\in {M}}\, g_{n}\sum_{1\le i\le d}\left(2f_{n}-f_{n+e_i}-f_{n-e_i}\right) \\
 =&\sum_{n\in {M}}v_n g_{n}f_{n}\,+\, \sum_{n\in {M}}\, 2d g_{n}f_{n}-\, \sum_{n\in {M}}\sum_{1\le i\le d}(g_nf_{n+e_i}+g_nf_{n-e_i}) \\
   =&\sum_{n\in {M}}\, \left(\frac{\sum_i (u_{n+e_i}+u_{n-e_i})}{u_n}+\frac{1}{u_n} \right)\, g_{n}f_{n}-\, \sum_{n\in {M}}\sum_{1\le i\le d}(g_nf_{n+e_i}+g_{n+e_i}f_n  )
\end{align*}
In the last step we used \eqref{eq:temp99} and the periodic condition
$\sum_{1\le n_i\le K}g_nf_{n-e_i}=\sum_{1\le n_i\le K}g_{n+e_i}f_n$.   
Continuing the computation, one has 
\begin{align*}
    &\sum_{n\in {M}}\, \left(\frac{\sum_i u_{n+e_i}+u_{n-e_i}}{u_n}+\frac{1}{u_n} \right)\, g_{n}f_{n}-\, \sum_{n\in {M}}\sum_{1\le i\le d}(g_nf_{n+e_i}+g_{n+e_i}f_n )    \\
      =&\sum_{n\in {M}}\, \frac{1}{u_n} \, g_{n}f_{n}
      +\sum_{n\in {M}}\,\sum_{1\le i\le d}\left( \frac{ u_{n+e_i}}{u_n} \, g_{n}f_{n}
       +\frac{ u_{n-e_i}}{u_n} \, g_{n}f_{n}\right) -\, \sum_{n\in {M}}\sum_{1\le i\le d}(g_nf_{n+e_i}+g_{n+e_i}f_n) \\
      =&\sum_{n\in {M}}\, \frac{1}{u_n} \, g_{n}f_{n}
      +\sum_{\mathclap{ \substack{n\in {M}\\1\le i\le d}}} \frac{ u_{n+e_i}}{u_n} \, g_{n}f_{n}
       +\sum_{\mathclap{ \substack{n\in {M}\\1\le i\le d}}} \frac{ u_{n}}{u_{n+e_i}} \, g_{n+e_i}f_{n+e_i}
       -\, \sum_{\mathclap{ \substack{n\in {M}\\1\le i\le d}}}(g_nf_{n+e_i}+g_{n+e_i}f_n )
\end{align*}
where we used again the boundary condition for $u,f,g\in \cH $ to rewrite the sum
$$\sum_{1\le n_i\le K} \frac{ u_{n-e_i}}{u_n} \, g_{n}f_{n}=\sum_{1\le n_i\le K} \frac{ u_{n}}{u_{n+e_i}} \, g_{n+e_i}f_{n+e_i}.$$ Therefore, 

\begin{align*}
 &\ipc{g}{H  f}_\cH \\
      =&\sum_{n\in {M}}\, \frac{1}{u_n} \, g_{n}f_{n} +\sum_{n\in {M}}\,\sum_{1\le i\le d} \frac{ u_{n+e_i}}{u_n} \, \left(g_{n}f_{n}
       + \frac{ u_{n}}{u_{n+e_i}} \, g_{n+e_i}f_{n+e_i}-g_nf_{n+e_i}-g_{n+e_i}f_n \right)\\
        =&\sum_{n\in {M}}\, \frac{1}{u_n} \, g_{n}f_{n} +\sum_{n\in {M}}\,\sum_{1\le i\le d} u_{n+e_i} u_n \left( \frac{ g_{n}}{u_n} \cdot   \frac{ f_{n}}{u_n}
       + \frac{ g_{n+e_i}}{u_{n+e_i}}\cdot \frac{ f_{n+e_i}}{u_{n+e_i}}
       - \frac{ g_{n}}{u_n} \cdot \frac{ f_{n+e_i}}{u_{n+e_i}}-\frac{ g_{n+e_i}}{u_{n+e_i}}\cdot \frac{ f_{n}}{u_n}  \right) \\
        =&\sum_{n\in {M}}\, \frac{1}{u_n} \, g_{n}f_{n} +\sum_{n\in {M}}\,\sum_{1\le i\le d} u_{n+e_i} u_n \left( \frac{ g_{n+e_i}}{u_{n+e_i}}- \frac{ g_{n}}{u_n} \right)\cdot \left( \frac{ f_{n+e_i}}{u_{n+e_i}}-  \frac{ f_{n}}{u_n}  \right) \\
         =&\sum_{n\in {M}}\, \frac{1}{u_n} \, g_{n}f_{n} +\sum_{n\in {M}}\,\sum_{1\le i\le d} u_{n+e_i} u_n\cdot  \nabla_i\frac{ g_{n}}{u_n}\cdot  \nabla_i \frac{ f_{n}}{u_n} 
\end{align*}

\end{proof}
Equation \eqref{eq:eff1} shows that $\ipc{f}{H  f}_\cH$ is conjugated to a new Hamiltonian which has a similar form but with a new potential $1/u_n$ replacing $v_n$. The new potential captures effects of both the kinetic and the potential energies.  Similar conjugate was first found for the continuous case in \cite{ADFJM-CPDE}, which plays the most important role. In the next section, we are going to show that  \eqref{eq:eff1} implies that eigenfunctions  have ``most'' of their mass in the region where $1/u_n$ is relatively small (effective wells); and exponential decay in the complementary region. Such result is usually referred as decay estimates of Agmon type. 

\subsection{Agmon estimates}\label{sec:pfAgmon}

To formulate our result precisely, we need a few more definitions. Suppose $(\mu,\varphi)$ is an eigenpair and $u$ is the landscape function of  $H $ in $\cH=\ell^2(M)$ , i.e., $H \varphi=\mu\varphi$, and $Hu=\vec 1$. For $\delta>0$, let the classical allowed region associated to the effective potential $1/{u_n}$ be given by
\begin{align}\label{eq:allowed-region}
    J_\delta(\mu):=\left \{ {n\in {M}} :\ \frac{1}{u_n}\le \mu+\delta \right\}
\end{align}

Define 
\begin{align}\label{eq:Agmon-weight}
    w_\mu(n)=\left(\frac{1}{u_n}-\mu\right)_+:=\max\left\{\frac{1}{u_n}-\mu,0\right\},\ n \in {M} 
\end{align}

The Agmon metric associated to this weight function will be defined as 
\begin{align}\label{eq:Agmon-metric}
    \rho_\mu(n,m)
    =\inf_{\gamma\in \cP(n,m)}\sum_{j=1}^{k(\gamma)}\ln \left(1+\sqrt{ \min\{w_\mu(\gamma_j),w_\mu(\gamma_{j+1})\}}\right)
\end{align}
where $\cP(n,m)$ is the collection of possible unit step path in $M\subset\Z^d$, connecting $n$ and $m$, 
\begin{align*}
&    \cP(n,m)=\left\{\, \gamma= \{\gamma_j\}_{j=1}^k\subset {M},   k=k(\gamma)\in \N \ | \right.  \\
& \qquad \qquad \qquad \qquad \qquad  \left. \gamma_1=n,\gamma_k=m,\, |\gamma_{j+1}-\gamma_j|_1=1,j=1,\cdots,k-1 \right\}
\end{align*}

It is easy to check that the infimum can always be attained by a non-self-intersecting path in the cube ${M}$ since there are only finitely many of them. Therefore, $ \rho_\mu(n,m)\le \rho_\mu(n,\ell)+\rho_\mu(\ell,m),\ \forall m,n,\ell\in {M}$, which indeed defines a semi-metric in the cube.  Also note the metric between points in a connected component of the set $\{w_\mu(n)=0\}\subset  J_\delta(\mu)$ is zero. 

The Agmon distance (weight) to the classical allowed region is defined as 
\begin{align}\label{eq:Agmon-dis}
    h_n=\inf_{m\in J_\delta(\mu)}\rho_\mu(n,m)
\end{align}
With these notations, we have
\begin{theorem}\label{thm:Agmon}
Suppose $K\in \N,K\ge 3$. Let $H=-\Delta+V $ be as in \eqref{eq:opH} and let $(\mu,\varphi)$ be an eigenpair of $H $. For $\delta>0,0<\mu\le  V_{\max}-\delta$, let $h_n$ be given as in \eqref{eq:Agmon-dis} for $\mu,\delta$ and $1/u_n$. There is an absolute constant $C>0$ such that for all $0<\alpha < 1/\sqrt{Cd}$, the following holds true: 
\begin{align}\label{eq:Agmon-loc}
    \sum_{h_n\ge 1}e^{2\alpha h_n}\varphi_n^2\, \le\, \frac{C_0} {\delta} \sum_{ {n\in {M}} }\varphi_n^2
\end{align}
where 
\begin{align}\label{eq:C0}
    C_0=\frac{4e^{2\alpha}d+(2+ 6C\alpha^2)e^{2\alpha} d\cdot V_{\max}}{1- Cd\alpha^2}
\end{align} 
\end{theorem}
\begin{remark}
 In the continuum setting, the Agmon metric and distance are usually defined in a slightly different way, see e.g. \cite{H}. Let $w(x)$ be a non-degenerate weight in $\R^d$. The Agmon metric in $\R^d$ associated to $w$ is given by 
 \begin{align*}
     \rho^{\rm conti}(x,y)=\inf_{\gamma}\int_0^1\sqrt{w(\gamma(t))}\, |\dot \gamma(t)|\, dt
 \end{align*}
 where the infimum taken over absolutely continuous paths $\gamma:[0,1]\to \R^d$ from $\gamma(0)=x$ to $\gamma(1)=y$. 
 
 If the weight function $w$ is small, the Agmon metric for the discrete case will have the same order as the Agmon metric for the continuous case since $\ln(1+\sqrt{w})\sim \sqrt{w}$. For large $w$, the logarithm in \eqref{eq:Agmon-metric} is very important. It is not only necessary for our proof of Theorem \ref{thm:Agmon} in the discrete setting, see also the technical Lemma \ref{lem:Agmon-gradient} below. Also, \eqref{eq:Agmon-loc} implies that, roughly speaking, the eigenfunction decays at a rate $|\varphi_n|\lesssim e^{-\alpha h_n}$. The logarithmic order of $h_n$ is consistent with the logarithmic order of the Lyapunov exponent for some famous tight-binding models, e.g. the almost Mathieu operator \cite{J,BJ} in the large coupling regime. The precise order of $h_n$ is also related to the so-called inverse participation ratio of a localized eigenfunction $\varphi$, $ {\rm IPR}(\varphi)=\left({\sum_{n\in M}\varphi_n^4}\right)/{\left(\sum_{n\in M}\varphi_n^2\right)^2}$. The analysis and the numerical experiment for the Agmon distance $h_n$, the Lyapunov exponent and the  inverse participation ratio are more complicated and subtle. We defer these to another paper. 
\end{remark}

To prove  Theorem \ref{thm:Agmon}, we need two more technical lemmas. The first one is a preliminary estimate of the discrete Agmon distance. 
\begin{lem}\label{lem:Agmon-gradient}
Suppose $n,m\in {M}$, and $|m-n|_1=1$, then
\begin{align}\label{eq:Agmon-gradient}
    |h_{m}-h_n|\le \ln \left(1+\sqrt{\min\{w_\mu(n),w_\mu({m})\}}\right)
\end{align}
As a consequence, there is an absolute constant $C>0$ such that for any $\alpha>0$,  $n,m\in {M}$, and $|m-n|_1=1$,
\begin{align}\label{eq:Agmon-exp-gradient}
    \left(e^{\pm(\alpha h_{m}-\alpha h_n)}-1\right)^2\, \le \, C\alpha^2  w_\mu(n).
\end{align}
\end{lem}

\begin{remark}
 The proof of \eqref{eq:Agmon-gradient} is straightforward from the triangle inequality of $\rho_\mu$ and the definition of $h_n$. \eqref{eq:Agmon-exp-gradient} follows from the fact $|e^{|x|}-1|\lesssim \max(|x|,e^{|x|})$ for all $x$.  We omit the detail of the proof. 
\end{remark}

The next lemma is a direct application of Lemma \ref{lem:eff} to the eigenpair $(\mu,\varphi)$, which plays the key role in the proof of the main theorem. 
\begin{lem}\label{lem:eff-eigen}
Suppose $(\mu,\varphi)$ is an eigenfunction of $H$ in $\cH=\ell^2(M)$. Let  $u\in \cH$ be the landscape function for $H $ given by Theorem \ref{thm:landscape}. Then for any $g\in \cH $, the following identity holds true:
\begin{align}\label{eq:eff-eigen}
    &\sum_{n\in {M}}\, \left(\frac{1}{u_n}-\mu\right)\,   \varphi^2_ng^2_n +\sum_{n\in {M}}\sum_{1\le i\le d}\, u_{n+e_i}u_{n}\left(\frac{g_{n+e_i}\varphi_{n+e_i}}{u_{n+e_i}}-\frac{g_n\varphi_n}{u_{n}}\right)^2\\
 \qquad   =&  \sum_{n\in {M}}\sum_{1\le i\le d}
 \varphi_{n+e_i}\varphi_{n}\left(g_{n+e_i}-g_n\right)^2 \nonumber 
\end{align}
As a consequence of the positivity of $u$, one has 
\begin{align}\label{eq:eff-eigen1}
    \sum_{n\in {M}}\, \left(\frac{1}{u_n}-\mu\right)\,   \varphi^2_ng^2_n \ 
     \le       \frac{1}{2}\sum_{n\in {M}}\,  \varphi^2_{n}\sum_{|m-n|_1=1}\left(g_{m}-g_{n}\right)^2
\end{align}

\end{lem}
\begin{proof}
We denote by $g^2\varphi$ the sequence $(g^2\varphi)_n=g^2_n\varphi_n$ for simplicity. \eqref{eq:eff} implies that 
\begin{align*}
\ipc{g^2\varphi}{H \varphi}=\sum_{n\in {M}}\,\sum_{1\le i\le d} u_{n+e_i} u_n\cdot  \nabla_i\frac{g^2_n\varphi_n}{u_{n}}\cdot  \nabla_i \frac{\varphi_{n}}{u_{n}}  \, + \,
\sum_{n\in {M}}\frac{1}{u_n}g^2_n \varphi^2_n
\end{align*}
Expanding the first term on the right hand side, one has 
\begin{align*}
&\sum_{n\in {M}}\,\sum_{1\le i\le d} u_{n+e_i} u_n\cdot  \nabla_i\frac{g^2_n\varphi_n}{u_{n}}\cdot  \nabla_i \frac{\varphi_{n}}{u_{n}}  \\
   = &\sum_{n\in {M}}\,\sum_{1\le i\le d} u_{n+e_i} u_n\left(\frac{g^2_{n+e_i}\varphi_{n+e_i}}{u_{n+e_i}}-\frac{g^2_n\varphi_n}{u_{n}}\right)\left(\frac{\varphi_{n+e_i}}{u_{n+e_i}}-\frac{\varphi_{n}}{u_{n}}\right)\\
 =&   \sum_{n\in {M}}\,\sum_{1\le i\le d}\, \left[\,  u_{n+e_i} u_n\,  \left(\frac{g^2_{n+e_i}\varphi^2_{n+e_i}}{u^2_{n+e_i}}+\frac{g^2_n\varphi^2_n}{u^2_{n}}\right)-
 \varphi_{n+e_i}\varphi_{n}\left(g^2_{n+e_i}+g^2_n\right)\,\right]\\
 =&   \sum_{n\in {M}}\sum_{1\le i\le d}\, u_{n+e_i}u_{n}\left(\frac{g_{n+e_i}\varphi_{n+e_i}}{u_{n+e_i}}-\frac{g_n\varphi_n}{u_{n}}\right)^2-\sum_{n\in {M}}\,\sum_{1\le i\le d}
 \varphi_{n+e_i}\varphi_{n}\left(g_{n+e_i}-g_n\right)^2
\end{align*}
On the other hand $H \varphi=\mu \varphi$ gives
$$\ipc{g^2\varphi}{H \varphi}= \ipc{g^2\varphi}{\mu \varphi}=\sum_{n\in {M}}\, \mu  g^2_n \varphi^2_n$$
\eqref{eq:eff-eigen} now follows directly by putting the above three equations together. 

To show \eqref{eq:eff-eigen1}, the only estimate needed is $\varphi_{n+e_i}\varphi_{n}\le \frac{1}{2}(\varphi_{n+e_i}^2+\varphi_{n}^2)$.  Then it is enough to shift the sum containing $\varphi^2_{n+e_i}$ using the periodic boundary condition for $\varphi,g\in \cH $:
\begin{align*}
    \sum_{n\in {M}}\,  \varphi^2_{n+e_i}\left(g_{n+e_i}-g_n\right)^2=\sum_{n\in {M}}\,  \varphi^2_{n}\left(g_{n}-g_{n-e_i}\right)^2
\end{align*}
Summing over all $n\in {M}$ and $1\le i \le d$, one has 
\begin{align*}
    \sum_{n\in {M}}\sum_{1\le i\le d}
 \varphi_{n+e_i}\varphi_{n}\left(g_{n+e_i}-g_n\right)^2\le & 
 \frac{1}{2}\sum_{n\in {M}}\sum_{1\le i\le d}
 \varphi_{n}^2\left(g_{n-e_i}-g_n\right)^2+\frac{1}{2}\sum_{n\in {M}}\sum_{1\le i\le d}
 \varphi_{n}^2\left(g_{n+e_i}-g_n\right)^2
\end{align*}
Then 
\begin{align*}
    \sum_{1\le i\le d}\left[\left(g_{n+e_i}-g_n\right)^2+\left(g_{n}-g_{n-e_i}\right)^2\right]
    =\sum_{|m-n|_1=1}\left(g_{m}-g_{n}\right)^2
\end{align*}
implies \eqref{eq:eff-eigen1} by dropping the positive term containing $u_{n+e_i}u_n$ on the left hand side of \eqref{eq:eff-eigen}.
\end{proof}

Now we are ready to prove our main result. 
\begin{proof}[Proof of Theorem \ref{thm:Agmon}]
For $0<\mu\le V_{\max}-\delta$,  
Lemma \ref{lem:Agmon-gradient} and Lemma \ref{lem:maxP} imply that for all $n,m\in {M}, |m-n|_1=1 $, 
\begin{align}\label{eq:h1}
|h_m-h_n|\le \sqrt{\left(\frac{1}{u_n}-\mu\right)_+}  \le  \sqrt{V_{\max}}\\
\left(e^{\pm(\alpha h_{m}-\alpha h_n)}-1\right)^2\, \le \, C\alpha^2 \left(\frac{1}{u_n}-\mu\right)_+  \le C\alpha^2 V_{\max}
\end{align}
where $C$ is the absolute constant in Lemma \ref{lem:Agmon-gradient}. 

Now we are going to apply Lemma \ref{lem:eff-eigen} to the following test function 
\begin{align}
   g_j= \left\{
\begin{array}{ll}
     h_je^{\alpha h_j}, & \ {\rm if}\ h_j<1 \\
      e^{\alpha h_j}, &  \ {\rm if}\ h_j\ge 1 \\
\end{array} 
\right.
\end{align}

We need to estimate $g_m-g_n$, for $|m-n|_1=1$, in the following three cases. Let $I_1=\{j\in{M}\, |\, h_j\ge1\},\ I_2=\{j\in{M}\, |\, h_j<1\}$ and $ B_1(n)=\{m\in{M}\, |\,|m-n|_1=1\} $

{\noindent \bf Case I: $n\in I_1, m\in I_1\cap B_1(n)$}. 
\begin{align*}
    \left(g_{m}-g_n\right)^2=e^{2\alpha h_{n}}\left(e^{\alpha h_{m}-\alpha h_{n}}-1\right)^2 \le  C\alpha^2\, e^{2\alpha h_{n}}\, \left(\frac{1}{u_n}-\mu\right)_+ 
\end{align*}

{\noindent \bf Case II: $n\in I_2, m\in I_2\cap B_1(n)$}. 
\begin{align*}
    \left(g_{m}-g_n\right)^2=&\left|(h_m-h_n)e^{\alpha h_m}+h_n(e^{\alpha h_m}-e^{\alpha h_n}) \right|^2\\
    \le & 2e^{2\alpha h_m}\left(h_m-h_n \right)^2
    +2h_n^2e^{2\alpha h_{n}}\left(e^{\alpha h_{m}-\alpha h_{n}}-1\right)^2\\
    \le & 2e^{2\alpha}V_{\max}+2e^{2\alpha}\cdot C\alpha^2V_{\max}=2e^{2\alpha}(1+ C\alpha^2)V_{\max}
\end{align*}

{\noindent \bf Case III: $n\in I_1, m\in I_2\cap B_1(n)$ or $n\in I_2, m\in I_1\cap B_1(n)$}.

Suppose $h_n\ge 1,h_m< 1$. 
\begin{align*}
    \left(g_{m}-g_n\right)^2=&\left|(h_m-1)e^{\alpha h_m}+(e^{\alpha h_m}-e^{\alpha h_n}) \right|^2\\
    \le & 2e^{2\alpha h_m}\left(h_m-1 \right)^2
    +2e^{2\alpha h_{m}}\left(1-e^{\alpha h_{n}-\alpha h_{m}}\right)^2\\
    \le & 2e^{2\alpha}+2C\alpha^2e^{2\alpha}\cdot V_{\max}
\end{align*}
The estimate for $h_n< 1,h_m\ge 1$ is the same. 

We now use these three bounds to estimate the right hand side of \eqref{eq:eff-eigen1}. Note that $\#(I_i\cap B_1(n))\le \#( B_1(n))\le 2d$ for $i=1,2$ and all $n\in {M}$. 
\begin{align*}
  \sum_{n\in {M}}\,  \varphi^2_{n}\sum_{|m-n|_1=1}\left(g_{m}-g_{n}\right)^2
= & \sum_{n\in I_1}\,  \varphi^2_{n}\sum_{m\in I_1\cap B_1(n)}\left(g_{m}-g_{n}\right)^2+\sum_{n\in I_2}\,  \varphi^2_{n}\sum_{m\in I_2\cap B_1(n)} \left(g_{m}-g_{n}\right)^2\\
&+\sum_{n\in I_1}\,  \varphi^2_{n}\sum_{m\in I_2\cap B_1(n)} \left(g_{m}-g_{n}\right)^2+\sum_{n\in I_2}\,  \varphi^2_{n}\sum_{m\in I_1\cap B_1(n)}\left(g_{m}-g_{n}\right)^2\\
\le & 2dC\alpha^2\, \sum_{h_n\ge 1}\,  \varphi^2_{n}\, e^{2\alpha h_{n}}\, \left(\frac{1}{u_n}-\mu\right)_+ \\
&+2d\left( 2e^{2\alpha}(1+ C\alpha^2)V_{\max}\,+4e^{2\alpha}+4C\alpha^2e^{2\alpha} V_{\max} \right)\sum_{n\in {M}}\,  \varphi^2_{n}
\end{align*}

On the other hand, note that $\{\frac{1}{u_n}\le \mu\}\subset \{\frac{1}{u_n}\le \mu+\delta\}=J_\delta(\mu)$. Therefore, $\frac{1}{u_n}\le \mu$ implies that $h_n=0<1$ and $g_n=h_ne^{\alpha h_n}=0$. Hence, 
\begin{align*}
    \sum_{n\in {M}}\, \left(\frac{1}{u_n}-\mu\right)\,   \varphi^2_n g^2_n
    = \sum_{n\in {M}}\, \left(\frac{1}{u_n}-\mu\right)_+\,   \varphi^2_n g^2_n 
    \ge &\sum_{h_n\ge 1}\, \left(\frac{1}{u_n}-\mu\right)_+\,   \varphi^2_n g^2_n \\
    =&\sum_{h_n\ge 1}\, \left(\frac{1}{u_n}-\mu\right)_+\,   \varphi^2_n e^{2\alpha h_n}
\end{align*}

Combing these estimates with \eqref{eq:eff-eigen1}, one has 
\begin{align*}
   (1- Cd\alpha^2)\, \sum_{h_n\ge 1}\, \varphi^2_{n}g^2_n\, \left(\frac{1}{u_n}-\mu\right)_+
   \le  C_1 d  \sum_{n\in {M}}\, \varphi^2_{n}
\end{align*}
where $C_1=4e^{2\alpha}+(2+ 6C\alpha^2)e^{2\alpha}V_{\max}$. 

Finally, it is easy to check that $h_n\ge 1 $ implies that $1/u_n> \mu+\delta$. Therefore, 
\begin{align*}
   (1- Cd\alpha^2)\delta\, \sum_{h_n\ge 1}\, \varphi^2_{n}e^{2\alpha h_n}\, \le 
    C_1d\sum_{n\in {M}}\, \varphi^2_{n}
\end{align*}
which immediately gives \eqref{eq:Agmon-loc}. 
\end{proof}

\subsection{Dual landscape for the high energy modes}\label{sec:dual-LL}
In this part, we are going to show that the theory
of the localization landscape also  explains the localization of high energy states oscillating see Figure \ref{fig:BWphi290}. 
To this end, we study the symmetric property of the discrete Schr\"odinger operator $H $. For $\varphi\in \cH=\ell^2({M})$, let 
\begin{align}\label{eq:dualU}
    \wt{\varphi}_n=(-1)^{s(n)}\varphi_n, \ n\in {M}
\end{align}
where $s(n)=\sum_{j=1}^d n_j$ for $n=(n_1,n_2,\cdots,n_d)\in \Z^d$.  Recall that $M=\Z^d/K\Z^d$ is the equivalent class and $\varphi\in \ell^2(M)$ implies the periodic condition \eqref{eq:bdry-peri} $\varphi_n=\varphi_{n+Ke_i}$. For $\wt{\varphi}$ given by \eqref{eq:dualU}, it is easy to check that 
\begin{align}
   \wt{\varphi}_{n+Ke_i}=(-1)^{s(n+Ke_i)}\varphi_{n+Ke_i}=(-1)^K  (-1)^{s(n)}\varphi_{n}=(-1)^K \wt{\varphi}_{n}
\end{align}
In other words, the  transformation defined in \eqref{eq:dualU} will  change the periodic boundary condition on the finite box ${M}$, unless $K$ is even. We will restrict to the periodic boundary condition in this part and only consider the case where $K$ is an even integer and $\wt{\varphi} \in \ell^2(M)$. 

Now suppose  $(\mu,\varphi)$ is an  eigenpair of $H=-\Delta+V$ in $\cH=\ell^2({M})$. By substituting $\wt{\varphi}_n=(-1)^{s(n)}\varphi_n$ into the equation 
\begin{align*}
    -\sum_{1\le i\le d}\left(\varphi_{n+e_i}+\varphi_{n-e_i}\right)\, +\, 2d\varphi_n+v_n\varphi_n=\mu \varphi_n
\end{align*}
and using the the fact $s(n\pm e_i)=s(n)\pm 1$, one has 
\begin{align*}
    \sum_{1\le i\le d}\left(\wt{\varphi}_{n+e_i}+\wt{\varphi}_{n-e_i}\right)\, +\, 2d\wt{\varphi}_n+v_n\wt{\varphi}_n=\mu \wt{\varphi}_n
\end{align*}
It is easy to check that the above equation can be rewritten as 
\begin{align}\label{eq:dual-eq}
  -\sum_{1\le i\le d}\left(\wt{\varphi}_{n+e_i}+\wt{\varphi}_{n-e_i}\right)\, +\, 2d\wt{\varphi}_n+(V_{\max}-v_n)\wt{\varphi}_n=(4d+V_{\max}-\mu)\wt{\varphi}_n
\end{align}
Equivalently, one has for $\wt{\varphi}\in \cH$ given by \eqref{eq:dualU},  
\begin{align*}
    (-\Delta +\wt V)\wt{\varphi}=\wt \mu \wt{\varphi}
\end{align*}
 where  $\wt{V}=\{\wt v_n\}_{n\in {M}}, \wt v_n=V_{\max}-v_n$ is a non-negative potential acting in $\cH$  in the same way as $V$, and 
 \begin{align}\label{eq:dual-ev}
     \wt{\mu}:=4d+V_{\max}-\mu
 \end{align}

In other words,  $(\mu,\varphi)$ is an eigenpair of $H $ iff $(\wt \mu,\wt{\varphi})$ given as above is an eigenpair of a dual operator $\wt{H} :=-\Delta +\wt V$, defined by the left hand side of \eqref{eq:dual-eq}.  

If $K$ is odd, then the dual state $\wt{\varphi}$ defined  in \eqref{eq:dualU} will satisfy the so-called anti-periodic (AP) boundary condition instead, which does not belong to $\ell^2(M)$. In this case, the dual operator and the dual eigenvalue equation turns out to be
\begin{align}
     H^{\rm AP}\wt{\varphi}=\wt \mu \wt{\varphi},\ \ \  H^{\rm AP}=-\Delta^{\rm AP}+\wt V
\end{align}
e.g., one can easily find the matrix representation of $-\Delta^{\rm AP}$ acting on $\R^3$: 
\begin{align*}
    -\Delta^{\rm AP}=\begin{pmatrix}
    2 & -1 & 1 \\
    -1 & 2 & -1 \\
    1 & -1 & 2 
    \end{pmatrix}
\end{align*}
Unfortunately, Schr\"odinger operator $H^{\rm AP}$ with anti-periodic boundary condition does not satisfy the maximal principle Lemma \ref{lem:maxP}, e.g., 
\begin{align*}
    \begin{pmatrix}
    2 & -1 & 1 \\
    -1 & 2 & -1 \\
    1 & -1 & 2 
    \end{pmatrix}\begin{pmatrix}
    -1 \\
    1 \\
    3 
    \end{pmatrix}=\begin{pmatrix}
    0 \\
    0 \\
    4 
    \end{pmatrix}
\end{align*}
So far, the localization landscape theory is only known to work for operators satisfying certain types of maximal principle. We do not plan to extend the framework  of the localization landscape theory to operators such as $H^{\rm AP}$.  Below, we will only restrict to cube ${M}$ with even size $K$, when we are working on dual landscape theory. Note that the original landscape in Theorem \ref{thm:Agmon} has no restriction on $K$.   

Now let us consider the dual operator $\wt H  $ for even $K$. Suppose $(\wt \mu,\wt{\varphi})$ is an eigenpair of $\wt H $. Note that  $0\le \wt V\le V_{\max}$. Theorem \ref{thm:landscape} implies there is a dual landscape $\wt u$ satisfying $(\wt H  \wt u)_n= 1$ correspondingly. Now we can define the Agmon distance $\wt h_n$ with respect to this dual landscape function $\wt u$ and $\wt \mu$ exactly in the same way as \eqref{eq:Agmon-dis}. More precisely, the dual effective well associated to $\wt u$ is defined as
\begin{align}\label{eq:well-dual}
    \wt J_\delta(\wt \mu):=\left \{ {n\in {M}} :\ \frac{1}{\wt u_n}\le \wt \mu+\delta \right\}=\left \{ {n\in {M}} :\ \frac{1}{\wt u_n}\le 4d+V_{\max}-\mu+\delta \right\}
\end{align}
And the Agmon metric $\wt \rho_{\wt \mu}(n,m)$ will be defined as in \eqref{eq:Agmon-metric}, associated to $\left(1/\wt u_n-\wt \mu\right)_+$. 

The dual Agmon distance (weight) to is defined as 
\begin{align}\label{eq:Agmon-dis-dual}
    \wt h_n=\inf_{m\in \wt J_\delta(\wt \mu)}\wt \rho_{\wt \mu}(n,m)
\end{align}

Theorem \ref{thm:Agmon} implies that for $\wt \mu\le V_{\max}-\delta$, 
\begin{align*}
     \sum_{\wt h_n\ge 1}e^{2\alpha \wt h_n}\wt{\varphi}_n^2\, \le\, \frac{C_0} {\delta} \sum_{ {n\in {M}} }\wt{\varphi}_n^2
\end{align*}
where $C_0$ is the same constant given in \eqref{eq:C0}. The transform \eqref{eq:dualU} does not change the square amplitude of the states, i.e., $\wt{\varphi}_n^2=\varphi_n^2$ where $(\mu,\wt{\varphi})$ is the eigenpair of the original operator $H $. Note that \eqref{eq:dual-ev} implies $\mu\ge 4d+\delta$. In summary, this dual landscape  $\wt u$ provides exponential decay estimates of Agmon type  near the top of the spectrum of $H  $:
\begin{theorem}\label{thm:Agmon-dual} 
 Let $K\in \N,K\ge 3$ is an even integer and let $H=-\Delta+V $ be the Schr\"odinger operator on $\ell^2({M})$, as defined in \eqref{eq:opH}. Let $\wt H=-\Delta+V_{\max}-V$ and let $\wt u$ be the dual landscape function satisfying $(\wt H \wt u)_n=1$. Suppose $(\mu,\varphi)$ is an eigenpair of $H $.  For $\delta>0,\mu\ge 4d+\delta$, let $\wt \mu=4d+V_{\max}-\mu$ and let $\wt h_n$ be defined as in \eqref{eq:Agmon-dis-dual} for $\wt u, \wt \mu,\delta$. For all $0<\alpha < 1/\sqrt{Cd}$, the following holds true: 
\begin{align}\label{eq:Agmon-loc-dual}
    \sum_{h_n\ge 1}e^{2\alpha \wt h_n}\varphi_n^2\, \le\, \frac{C_0} {\delta} \sum_{ {n\in {M}} }\varphi_n^2
\end{align}
where $C,C_0$ are the same constants as in Theorem \ref{thm:Agmon}.

\end{theorem}

\section{Dirichlet boundary condition}\label{sec:diri}
In this Section, we discuss the generalization of the results to the Dirichlet boundary condition.  Let $M= [1,K]^d\cap\Z^d$ be a cube in $\Z^d$, with $K\in\N$ points in each direction. We denote by
\begin{align}\label{eq:o-bdry}
   \partial M=\left\{\ n\in\Z^d:\  n_i=0\ {\rm or}\  K+1\ \  \textrm{for (only) one }\ 1\le i\le d\ \right\} 
\end{align}
 the  outer boundary of ${M}$, and denote by 
 \begin{align}\label{eq:i-bdry}
   \partial^\circ M=\left\{\ n\in\Z^d:\  n_i=1\ {\rm or}\  K\ \textrm{for some}\ 1\le i\le d\ \right\} 
\end{align}  the inner boundary. We also denote by $\overline M= M\cup \partial M$. See  a cube of size $K=5$ in $\Z^2$ and its outer and inner boundaries in Figure \ref{fig:DisBound} for example.
\begin{figure}
    \centering
    \includegraphics[width=0.5\textwidth]{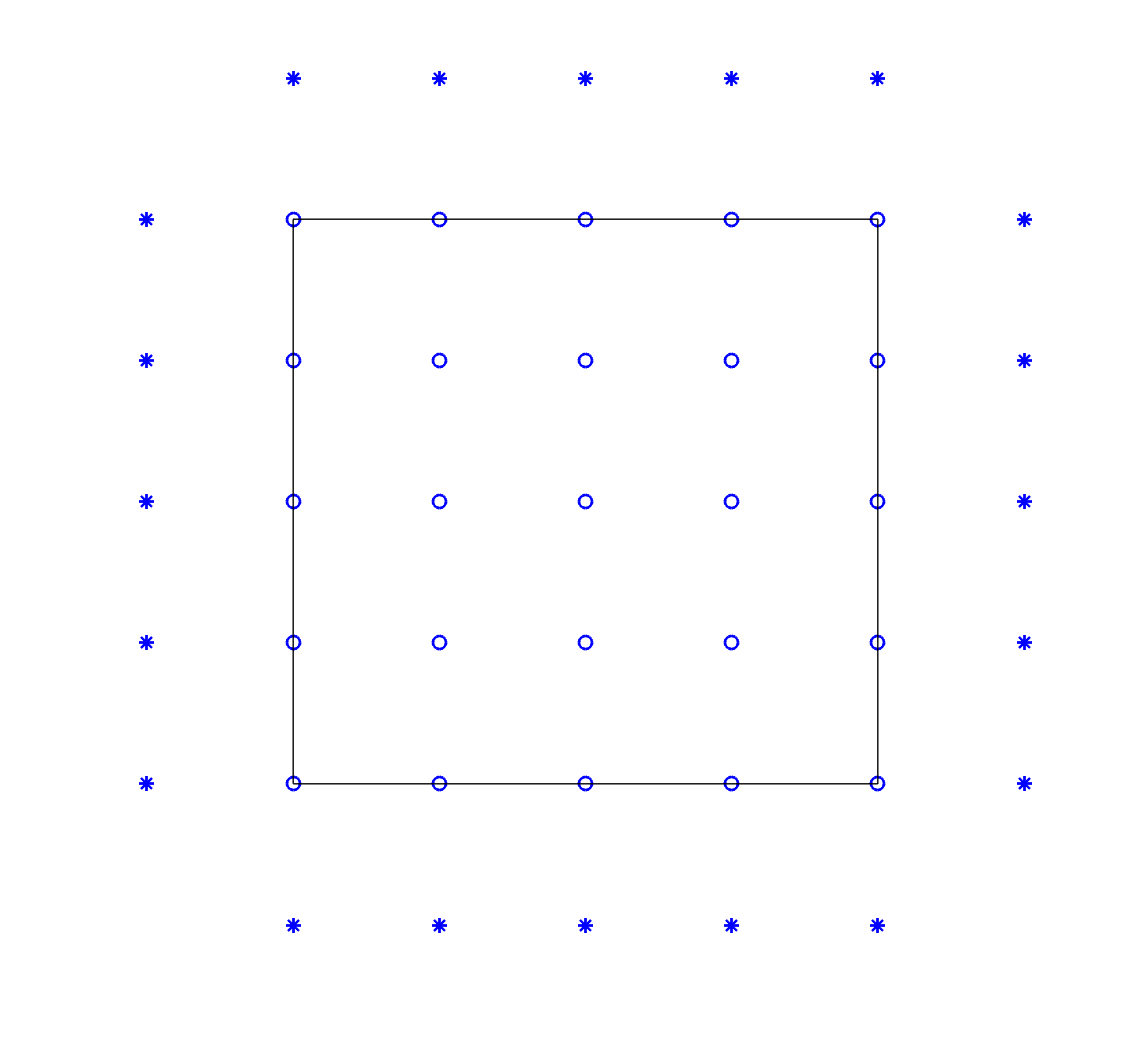}
    \caption{The collection of the circled dots is the box $M=\{1,2,3,4,5\}\times\{1,2,3,4,5\} \subset \Z^2$. The collection of the stared dots is the outer boundary of $M$. The circled dots connected by the straight line form the inner boundary of $M$.}
    \label{fig:DisBound}
\end{figure}

Let $\Delta$ be the Laplacian on $\ell^2(\Z^d)$ as defined in \eqref{eq:Lap}. We denote by $\Delta^{\rm D}$ the Laplacian on $\ell^2(M)$, with Dirichlet boundary condition, given by
\begin{align}\label{eq:Lap-diri}
   (\Delta^{\rm D}\phi)_n=\begin{cases}\sum\limits_{|m-n|_1=1}\phi_m\, -\, 2d\phi_n, \ \ & n\notin \partial^\circ M \\
   \sum\limits_{|m-n|_1=1,m\in M}\phi_m\, -\, 2d\phi_n, \ \ & n\in \partial^\circ M
   \end{cases}
\end{align}
If one consider the `extended' vector space 
\begin{align}
    \overline{\cH}=\left\{\, \bar \phi=(\phi_n)_{n\in \overline{M}}  \ \left| \right. \  (\phi_n)_{n\in M}\in \ell^2(M),\ \  
    (\phi_n)_{n\in \partial M}=0
    \right\}
\end{align}
then clearly, $(\Delta^{\rm D}\overline{\phi})_n=\sum_{|m-n|_1=1}\overline \phi_m\, -\, 2d\overline \phi_n, n\in M$. In other words, $\Delta^{\rm D}$ is the direct restriction of $\Delta$ on the cube $M$ by assigning zero value outside of $ M$, where the term in $\Delta$ connecting $M$ and $M^C$ are completely removed. This is one of the reasonable analogues \footnote{Such boundary condition is sometimes referred as the `simple' boundary condition in the context of tight-binding Schr\"odinger operators, which does not satisfy the so-called  Dirichlet-Neumann bracketing as in the continuous case.  We are not using Dirichlet-Neumann bracketing in the current work and therefore are not extending the discussion on this. For other boundary conditions in the tight-binding model, we refer readers to e.g. \cite{Kir}. } of the Dirichlet boundary condition in $\R^d$. Similar to the periodic case,  $-\Delta^{\rm D}$ acting on the finite dimensional space $\ell^2(M)$ has its matrix representation. Take $M=\{1,2,\cdots,K\}\subset\Z^1$ for example,  $\Delta^{\rm D}$  can be identified as the following $K\times K$ matrix, acting on $(\phi_1,\cdots,\phi_K)^T$, 
\begin{align}\label{eq:LapDiri}
    -\Delta^D=\begin{pmatrix}
	2 & -1 &0  & \cdots &  0 \\
	-1 &   2 &\ddots   &  \vdots \\
	0 & \ddots& \ddots&\ddots  & 0 \\
	\vdots  &\ddots & \ddots&  2   &-1 \\
	0  &\cdots & 0&-1 & 2
\end{pmatrix}_{K\times K}
\end{align}
The Schr\"odinger operator that we are interested is given by $H^{\rm D}=-\Delta^{\rm D}+V$, where $V$ is the same multiplication operator as in the periodic case \eqref{eq:opH}. The framework in the periodic case allows one to obtain similar  Agmon type of localization results for $H^{\rm D}$, with some slightly different technicalities. We will briefly explain the difference for the Dirichlet case below.  

We use the same notations of directional derivatives and gradient as in Section \ref{sec:LL} for vectors in $\ell^2(\Z^d)$. Restricting on a finite box,  one can easily check that the Green's identity \eqref{eq:Dform} in the Dirichlet case looks slightly different since the boundary is polluted.  Actually, for $f,\in \cH=\ell^2(M)$, it is more convenient to consider their extension $\bar f\in \overline{\cH}$ where $\bar f|_{M}=f$ and $\bar f|_{\partial M}=0$.  By the same computation as in Lemma \ref{lem:Dform}, one has
\begin{align}\label{eq:Dform-Diri}
    \ipc{g}{-\Delta^{\rm D} f}_{\cH}  =
    \sum_{\mathclap{ \substack{ n\in\partial^\circ{M}\ {  \rm and}\\ n\pm  e_i\in \partial {M}   }}}g_nf_n+\sum_{\mathclap{ \substack{n\in M\ {\rm and}\\ n+e_i\in {M}}}}(\nabla  g \cdot \nabla  f)(n) 
    =\sum_{\mathclap{ \substack{n\in {M}\ {\rm or}\\ n+e_i\in {M}}}}(\nabla \bar g \cdot \nabla \bar f)(n)  
\end{align}
Clearly,  $-\Delta^{\rm D}$ is strictly positive since $\ipc{f}{-\Delta^{\rm D}f}>0$, for any $0\neq f\in \ell^2(M)$.  Therefore $H^{\rm D}=-\Delta^{\rm D}+V$ is invertible on $\ell^2({M})$ since $V$ is non-negative, and the landscape function is well defined through the equation $(H^{\rm D}u)_n=1,n\in M$. The positivity of $u_n$ is again from the maximum principle for $H^{\rm D}$, whose  proof is similar to Lemma \ref{lem:maxP}. We omit the details here.  Now let $\vec 1\in \cH$ be the constant one vector as usual. Let $f=u-\beta^{-1} \cdot \vec 1$, where
\begin{align}\label{eq:beta}
    \beta=V_{\max}+d
\end{align}
It is easy to check that  $(\Delta^{\rm D}\vec 1)_n=0$ for $n\in M\backslash \partial^\circ {M}$, which implies that 
\begin{align*}
    (H^{\rm D} f)_n=1-\beta^{-1}   \cdot v_n\ge0, \ \ n\in {M}\backslash \partial^\circ {M}.
\end{align*}

If $n\in \partial^\circ {M}$, by the second line in \eqref{eq:Lap-diri} , one has
\begin{align*}
    -(\Delta^{\rm D}\vec 1)_n=-\left(\sum_{|m-n|_1=1,
    m\in  {M}}1\right)\, +2d =&2d-\# \left\{ m\in  {M}, |m-n|_1=1    \right\}\\
    =& \,\# \left\{ m\notin  {M}, |m-n|_1=1    \right\}:=k_n
\end{align*}
It is easy to check that for $n\in \partial^\circ {M}$, $1\le k_n \le d$. Therefore, 
\begin{align*}
      (H^{\rm D} f)_n=&(H^{\rm D}u)_n-\beta^{-1}\left(-(\Delta^{\rm D}\vec 1)_n+(V \vec 1)_n\right)\\
      =&1-\beta^{-1}   \cdot (k_n+v_n)\ge 1-\beta^{-1}  \cdot (d+V_{\max})\ge 0, \ \ n\in  \partial^\circ {M}.
\end{align*}
By the maximal principle, $f_n=u_n-\beta^{-1}  \ge 0$ for all $ {n\in {M}} $, which gives a slightly modified lower bound compared with \eqref{eq:u-lower}:
\begin{align}\label{eq:u-lower-diri}
    \min_{n\in M}u_n\ge \beta^{-1}=\frac{1}{V_{\max}+d}
\end{align}

Using the Green's identity \eqref{eq:Dform-Diri}, one can obtain uncertainty principle Lemma \ref{lem:eff} and Lemma \ref{lem:eff-eigen} for the Dirichlet case in the same way, by setting the terms to be zero if $n+e_i\notin M$. We summarize the results as the following lemma and omit the proof. 
\begin{lem}
Suppose $(H^{\rm D}  u)_n=1,n\in {M}$ for $u\in \cH=\ell^2(M)$.  For any $f,g\in \cH $,  one has
\begin{align}
    \ipc{g}{H^{\rm D}  f}_\cH
         =&\sum_{\mathclap{ \substack{1\le i\le d \\ n\ {\rm and}\ n+e_i\in {M}}}} \left(u_{n+e_i} u_n\cdot  \nabla_i\frac{ g_{n}}{u_n}\cdot  \nabla_i \frac{ f_{n}}{u_n} \right)\,+\, \sum_{n\in {M}}\, \frac{1}{u_n} \, g_{n}f_{n}
\end{align}

If, in addition, $(\varphi,\mu)$ is an eigenpair of $H^{\rm D}$ on $\cH$, then for any $g\in \cH$, the following is true:
\begin{align}\label{eq:eff-eigen2-diri}
    \sum_{n\in \Lambda^d}\, \left(\frac{1}{u_n}-\mu\right)\,   \varphi^2_ng^2_n \ 
     \le       \frac{1}{2}\sum_{n\in \Lambda^d}\,  \varphi^2_{n}\sum_{m\in \Lambda^d,  |m-n|_1=1}\left(g_{m}-g_{n}\right)^2
\end{align}
\end{lem}
With this lemma, one can obtain the Agmon type of estimates for the eigenfunction $\varphi$ of $H^{\rm D}$. The proof is exactly the same as for Theorem \ref{thm:Agmon}, except that one needs to replace the upper bound of $1/u_n\le V_{\max}$ by $V_{\max}+d$, as provided by \eqref{eq:u-lower-diri} in the Dirichlet case. 

For the high energy state, dual landscape works in the Dirichlet case exactly in the same way as for the periodic case. Moreover, following the computation in \eqref{eq:dualU}-\eqref{eq:dual-ev}, one can easily check that the zero boundary condition is preserved for the dual operator for any $M$. In other words, dual landscape in the Dirichlet case does not need the size of the domain $M$ to be even. For any $K\in \N$,  $(\mu,\varphi)$ is an eigen-pair for $H^{\rm D}$, iff $(4d+V_{\max}-\mu,\wt \varphi)$ is  an eigen-pair for $\wt H^{\rm D}$, where $\varphi,\wt \varphi$ are linked through \eqref{eq:dualU}  and the operator $\wt{H}^{\rm D}:=-\Delta^{\rm D}+V_{\max}-V$.

In summary, let $h_n$ be the Agmon distance defined as in \eqref{eq:Agmon-weight} for $H^{\rm D}$ associated to $u_n$ and let $\wt h_n$ be defined for the dual operator  $\wt H^{\rm D}$ associated to $\wt u_n$ as in \eqref{eq:Agmon-dis-dual}, one has 
\begin{theorem}
 Suppose $K\in \N,K\ge 3$. Let $H^{\rm D}=-\Delta^{\rm D}+V $ be as in \eqref{eq:Lap-diri}. Let $(\mu,\varphi)$ be an eigenpair of $H^{\rm D} $.  There is an absolute constant $C>0$ such that for all $0<\alpha < 1/\sqrt{Cd}$, the following holds true: for any $\delta>0$, if $0<\mu\le  V_{\max}-\delta$, then 
\begin{align}\label{eq:Agmon-loc-diri}
    \sum_{h_n\ge 1}e^{2\alpha h_n}\varphi_n^2\, \le\, \frac{C_2} {\delta} \sum_{ {n\in {M}} }\varphi_n^2
\end{align}
and if  $\mu\ge 4d+\delta$, then
\begin{align}\label{eq:Agmon-loc-dual-diri}
    \sum_{h_n\ge 1}e^{2\alpha \wt h_n}\varphi_n^2\, \le\, \frac{C_2} {\delta} \sum_{ {n\in {M}} }\varphi_n^2
\end{align}
where 
\begin{align}\label{eq:C1}
    C_2=\frac{4e^{2\alpha}d+(2+ 6C\alpha^2)e^{2\alpha} d\cdot (V_{\max}+d)}{1- Cd\alpha^2}
\end{align} 
\end{theorem}

\section{More numerical experiments}\label{sec:num}
Our approach in Section \ref{sec:Agmon} is deterministic. In this part, we mainly present numerical experiments  
to verify our theoretical results. We consider the Schr\"odinger operators with disordered potentials for $d=1,2$
to show how the structure of wells and barriers of the effective 
potential $W_n=1/u_n$ can be used to predict the locations and approximate supports of localized eigenfunctions. The numerical behavior of the periodic case and the Dirichlet case  are similar.  We will show one example of the ground state for the 1-d periodic boundary condition first.

Our main result Theorem \ref{thm:Agmon} says that for an eigenpair $(\mu,\varphi)$, the component $\varphi_n$ is at most
of size $e^{c_1 h_n}$; with $h_n$ the effective distance from a union of effective potential wells
$J_\delta(\mu):=\left \{ {n\in {M}} :\ 1/{u_n}\le \mu+\delta \right\}$. 
For general deterministic potential, the theorem does not guarantee  the number or the size of  wells (connected component in $J_\delta(\mu)$) as well as the separation 
of the wells. In other words, $\varphi_n$ only decays insofar as the function $h_n$ grows. But numerical results have shown great performance, in which 
the growth of $h_n$ and the way it matches the decay of eigenfunctions are very evident. 
To illustrate this, we still consider the Bernoulli potential in Figure \ref{fig:BV}, based on which, we compare the location of the localized first eigenfunction 
$\phi^{(1)}$ (in green) and its relation to the effective potential $W_n=1/u_n$, displayed in Figure \ref{fig:BWphi1}. 

\begin{figure}[H]
	\centering
	\includegraphics[width=\textwidth]{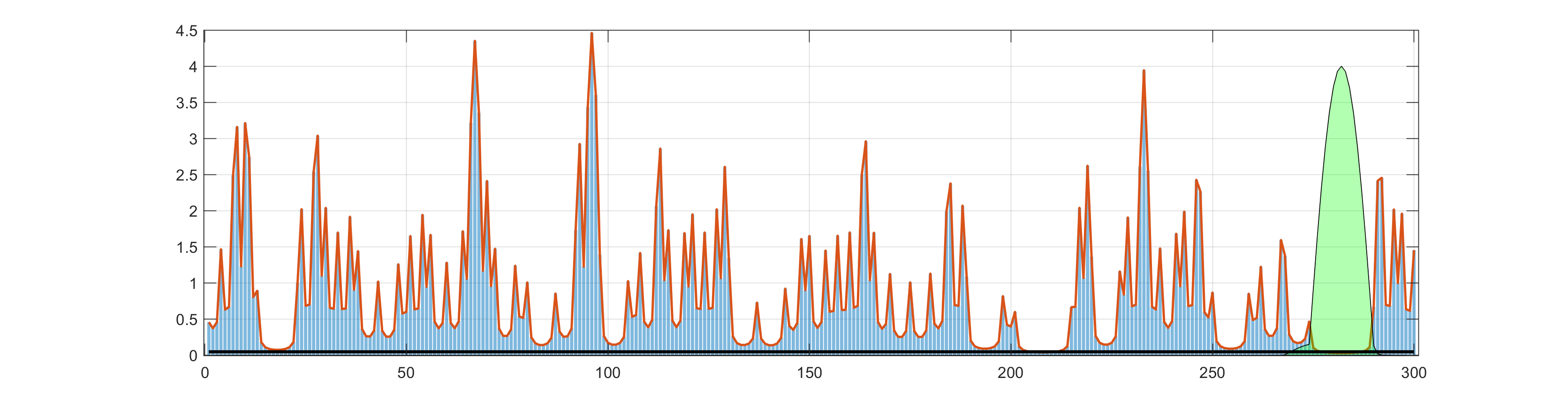}
	\includegraphics[width=\textwidth]{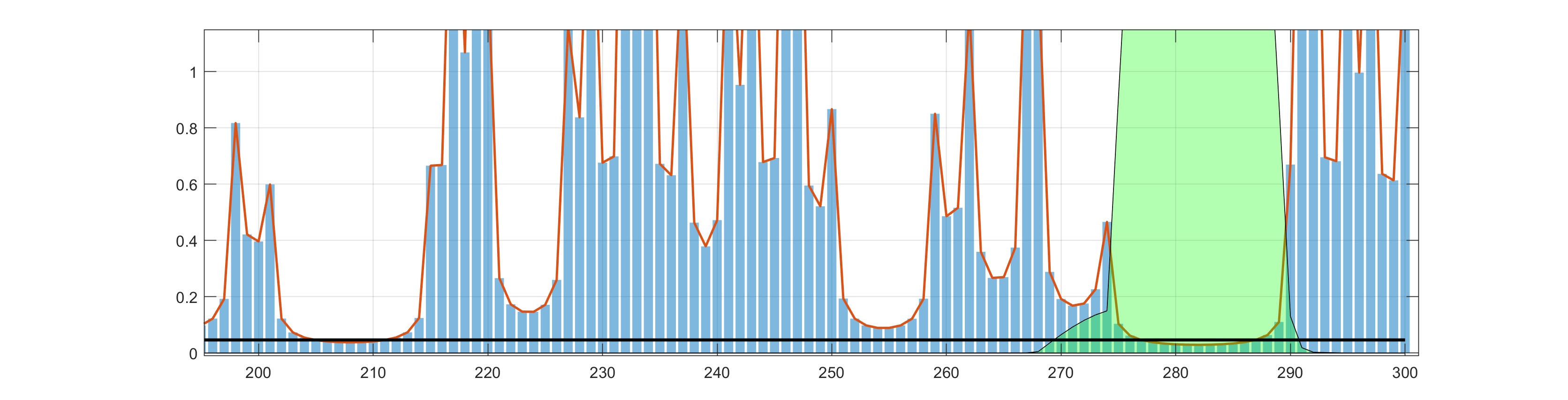}
	\caption{Effective potential $1/u_n$ (filled blue piecewise constant), and the effective wells $J_\delta(\mu_1)=\left \{ 1/u_n\le \mu_1+\delta \right\}$ with the first eigenfunction $\phi^{(1)}$ superimposed in green scale, for the operator in \eqref{eq:LapPeri} with the  Bernoulli random potential in Figure \ref{fig:1DBV}. The horizontal line segment near the bottom indicates the value $\mu_1+\delta$, where $\mu_1=0.0367$ and $\delta=0.01$.}
	\label{fig:BWphi1}
\end{figure}	
In Figure \ref{fig:BWphi1}, the horizontal line, $\mu_1+\delta$ 
determines a single well $J_\delta(\mu_1)=\left \{ 1/u_n\le \mu_1+\delta \right\}$.
Note that the the fundamental state is almost trapped in this single well actually, which shows the effectiveness of the landscape function.

\begin{figure}[H]
    \centering
    \includegraphics[width=\textwidth]{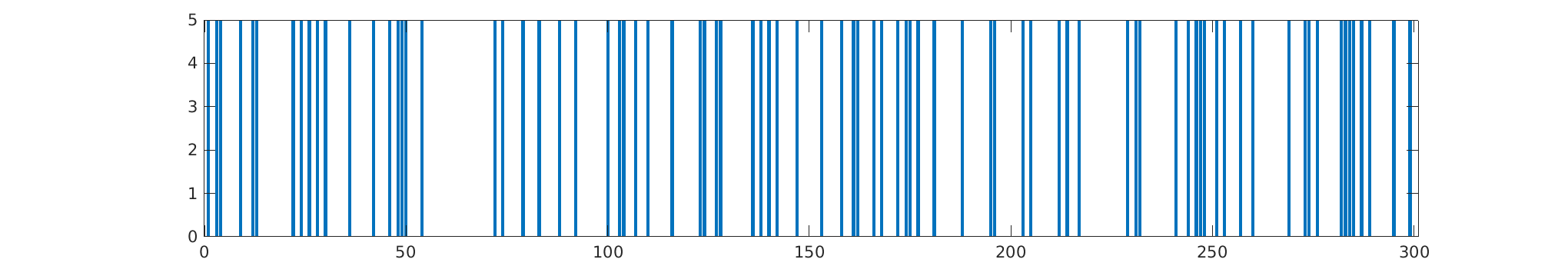}
    \caption{A different realization of i.i.d. $\{0,5\}$ Bernoulli random potential with  probabilities $70\%$ and $30\%$, over a lattice of size $300$.}
    \label{fig:1DBV}
\end{figure}

All the remaining numerical results will be for the Dirichlet boundary condition as in \eqref{eq:Lap-diri}. We start again with a different realization 1-d Bernoulli random potential over a lattice of size $K=300$, see Figure \ref{fig:1DBV}. We  consider the eigenpair in the middle other than the ground state this time. Actually, we plot 
the twelfth eigenstate $\phi^{(12)}$ in Figure \ref{fig:Bphi12} for the random potential in Figure \ref{fig:1DBV}. 
\begin{figure}[H]
	\centering
	\includegraphics[width=\textwidth]{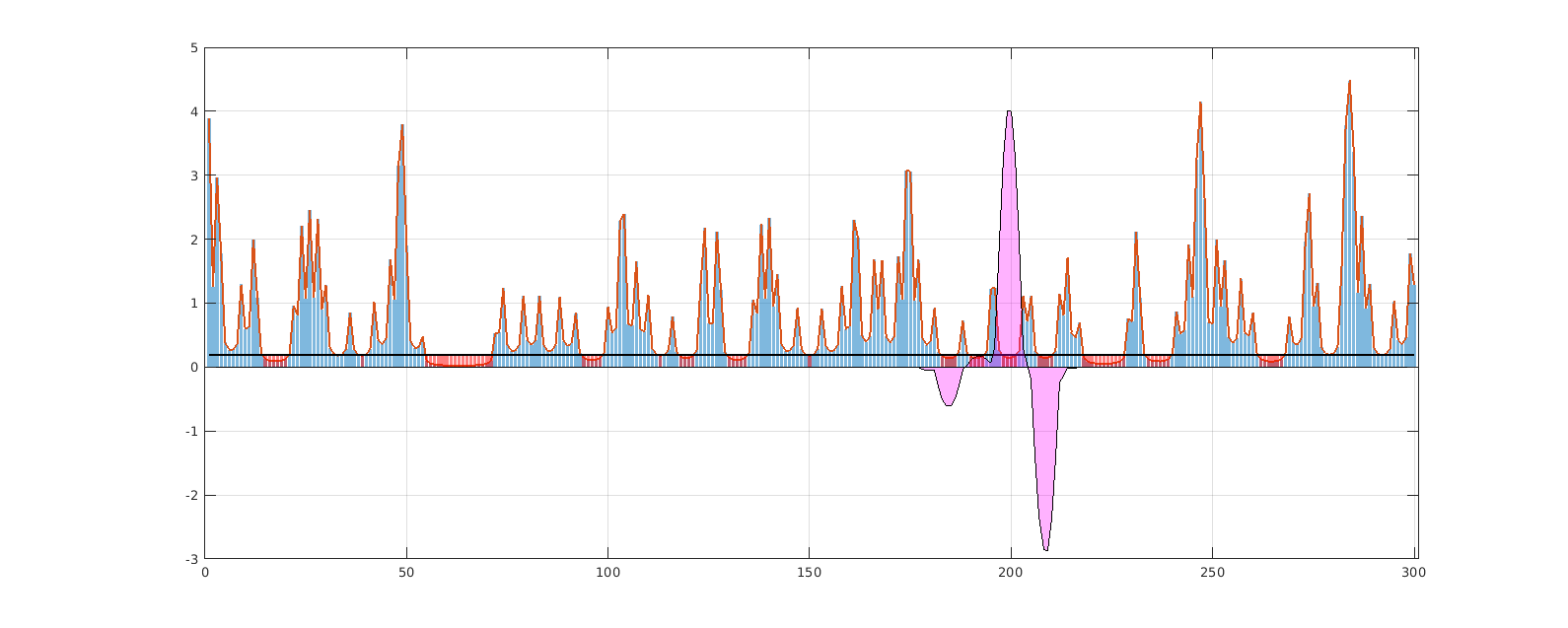}
	\caption{The effective wells $J_\delta(\mu_{12})=\left \{ 1/u_n\le \mu_{12}+\delta \right\}$ (filled red) with the twelfth eigenfunction $\phi^{(12)}$ superimposed in pink scale, for the Bernoulli random potential in Figure \ref{fig:1DBV}. The horizontal line segment indicates the value $\mu_{12}+\delta$, where $\mu_{12}=0.1805$ and $\delta=0.01$.}
	\label{fig:Bphi12}
\end{figure}
Different from the single well of the ground state in  Figure \ref{fig:BWphi1}, $J_\delta(\mu_{12})$ (in red) is now a union of several effective wells. 
The twelfth eigenfunction $\phi^{(12)}$ is superimposed in  pink scale.   
The result shows that  most of $\phi^{(12)}$ could also be captured by a cluster of components of $J_\delta(\mu_{12})$ effectively.  
In fact, dozens of eigenfunctions coincide essentially with single components or clusters of components.

Further, we  move to high energy modes using the dual landscape theory. In the following Figure \ref{fig:BWphi290},
the high energy state $\phi^{(290)}$ presents highly spatial oscillations, as predicted by 
the transform \eqref{eq:dualU}.  $\wt J_\delta(\mu_{290})$ here clearly signals the subregion
of localization, by a cluster of wells. 
\begin{figure}
	\centering
	\includegraphics[width=\textwidth]{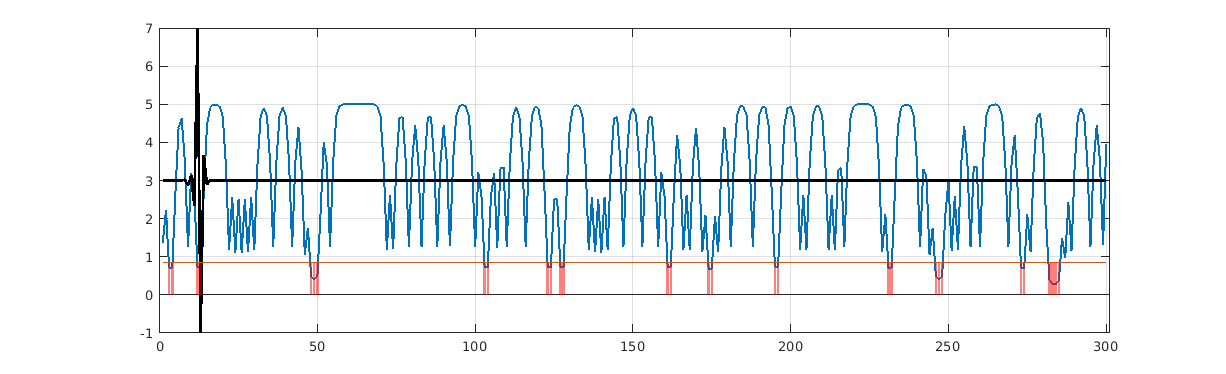}
	\caption{Localization and dual landscape near the top of the spectrum: The $290^{\rm th}$ eigenfunction $\phi^{(290)}$ is plotted in black for the same potential of Figure \ref{fig:1DBV}. For better
		visibility, the eigenstate is vertically shifted (three units up).The dual effective potential, plotted in blue scale, is given by the dual landscape $1/\wt u_n$.  The dual effective wells $\wt J_\delta(\mu_{290})=\left \{ 1/u_n\le \wt \mu_{290}+\delta \right\}$, are plotted in red. The horizontal line segment indicates the value $\wt \mu_{290}+\delta$,  where $\wt \mu_{290}=9-\mu_{290}=9-8.1670$ is the dual  and $\delta=0.01$.}
	\label{fig:BWphi290}
\end{figure}

Next, we will extend our numerical experiments to the uniformly random potential cases in both  one dimension and two dimension (we still consider the Dirichlet boundary condition).

We first consider the 1-d case and take the fourth eigenfunction $\phi^{(4)}$ as an example.
The uniform random potential $V$, with $V_{max}=5$ and $\phi^{(4)}$ are plotted in Figure \ref{fig:1DUnif}. 
Obviously, the the prediction of localized eigenfunction using $1/u_n$ is as effective as those in the Bernoulli case. 

\begin{figure}[H]
	\centering
	\includegraphics[width=\textwidth]{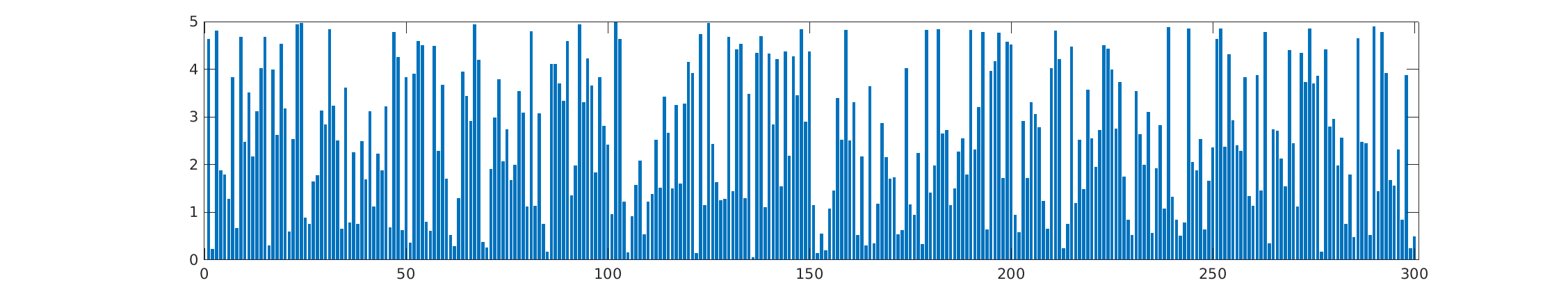}
	\includegraphics[width=\textwidth]{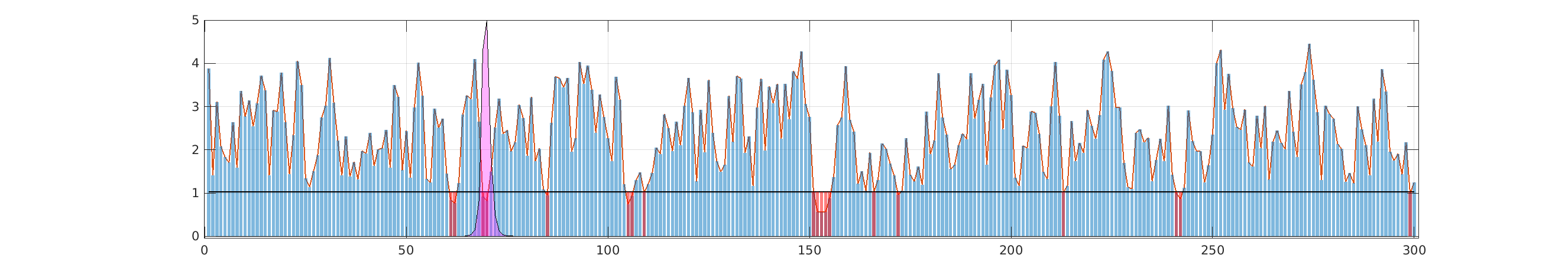}
	\caption{Top row: a random potential with uniform distribution over a lattice of size $300$. Bottom row: the effective wells $J_\delta(\mu_4)=\left \{ 1/u_n\le \mu_4+\delta \right\}$ (filled red) with the fourth eigenfunction $\phi^{(4)}$ superimposed in pink scale, for the random potential in the top row, where $\mu_4=1.0183$ and $\delta=0.01$.}
	\label{fig:1DUnif}
\end{figure}	

We also compare the separation of the effective wells for different $V_{\max}$. The $50^{\rm th}$  eigenfunction $\phi^{(50)}$ and the associated effective wells for uniform random potentials with different $V_{\max}$ are plotted  in Figure \ref{fig:sepWells}. It is clear that as $V_{\max}$ increases, the separation of the wells becomes more effective. 
\begin{figure}
	\centering
	\includegraphics[width=\textwidth]{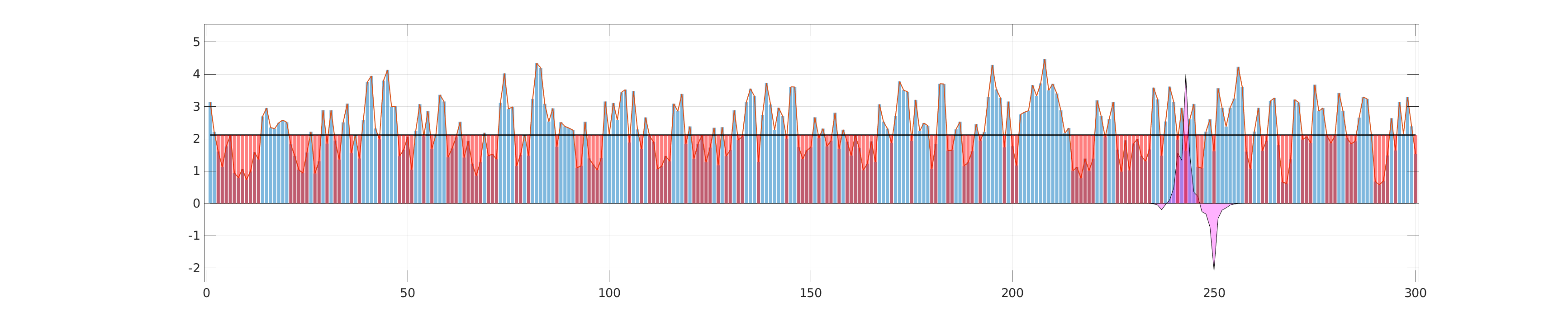}
	\includegraphics[width=\textwidth]{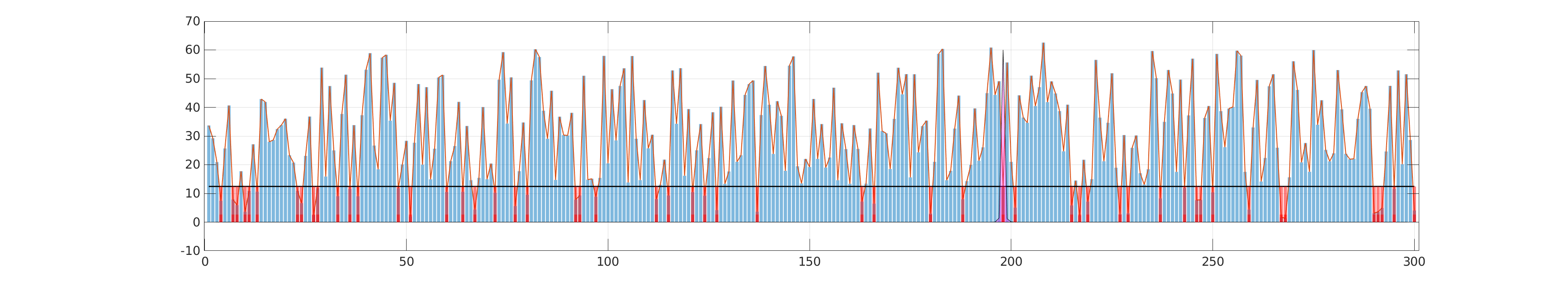}
	\caption{The $50^{\rm th}$  eigenfunction $\phi^{(50)}$ (superimposed in pink scale) and the associated effective wells (filled red) for uniformly random potentials with different $V_{\max}$. Top row: $V_{\max}=5$. Bottom row: $V_{\max}=64$.}
	\label{fig:sepWells}
\end{figure}

In 2-d case, we consider a square $M=\{1,\cdots,100\}\times \{1,\cdots,100\}\subset \Z^2$,
with a random potential $v_n$ chosen uniformly in $[0,5]$ for $n\in M$. 
We  compute  the effective potential $W_n=1/u_n$ for this random potential, see Figure \ref{fig:2dW}. 

\begin{figure}[H]
	\includegraphics[width=0.7\textwidth]{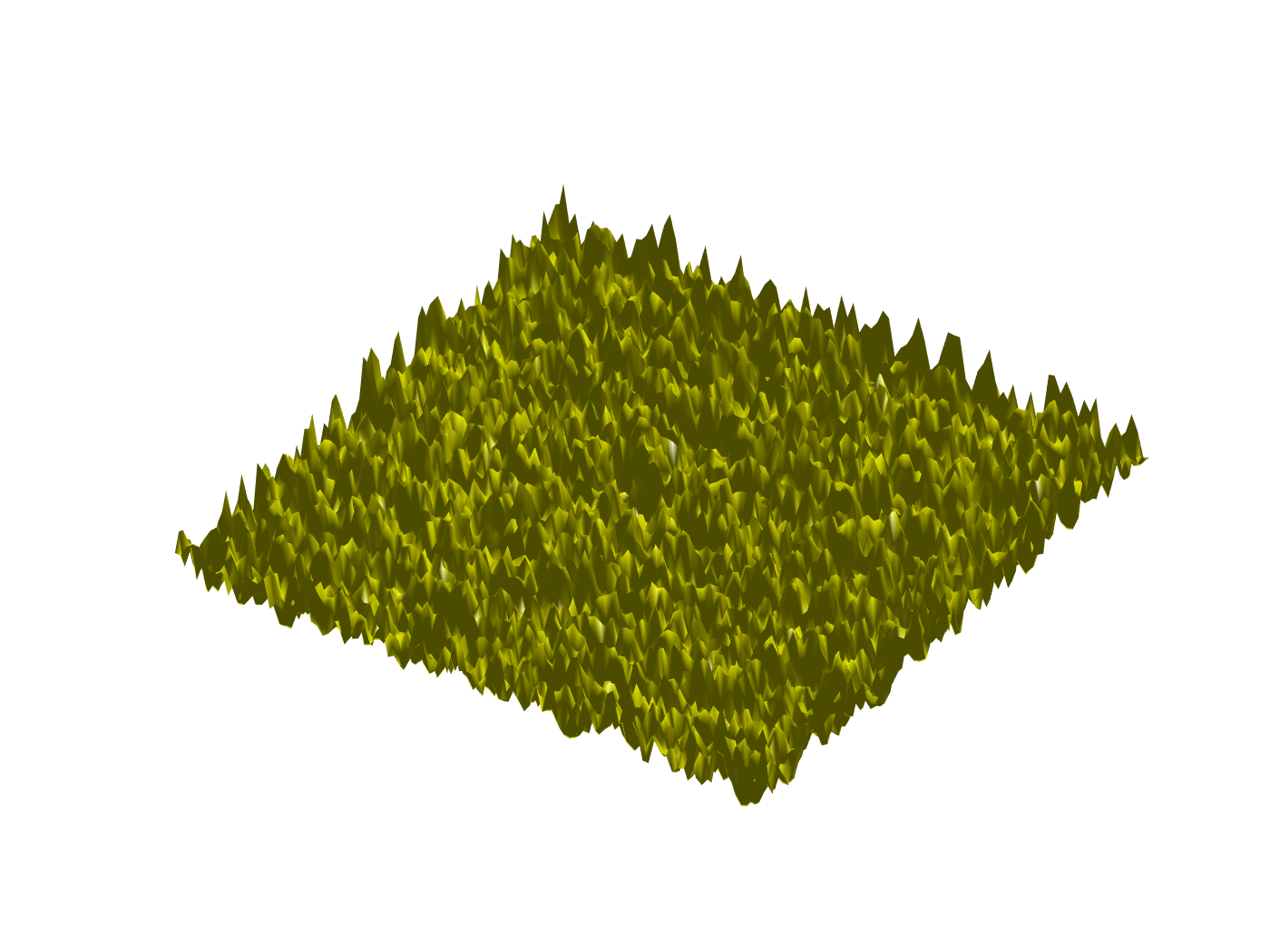}
	\caption{The effective potential associated to a  uniform random potential in $\Z^2$. It is shown with its crestlines which partition the domain into a few hundred basins of attraction surrounding wells.}
	\label{fig:2dW}
\end{figure}

\begin{figure}[H]
	\centering
	\subfigure[]{
		\includegraphics[width=0.5\textwidth]{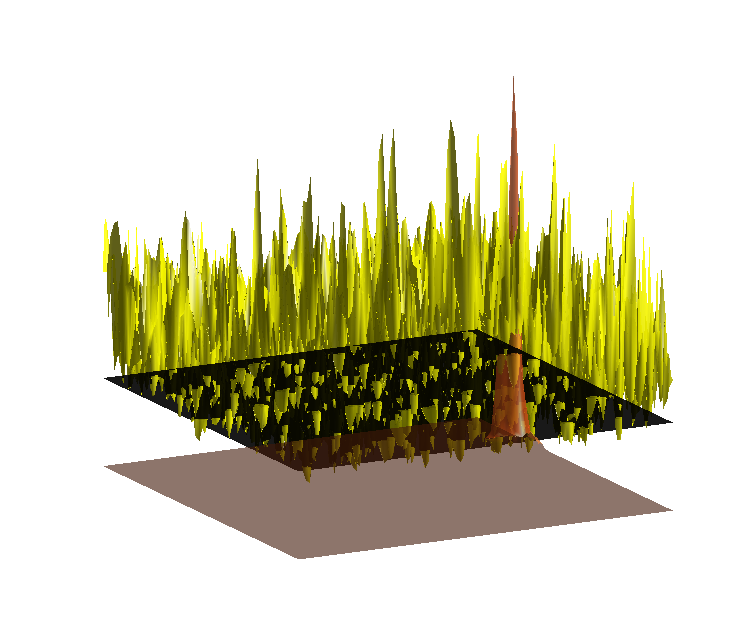}
	}
	\subfigure[]{
		\includegraphics[width=0.4\textwidth]{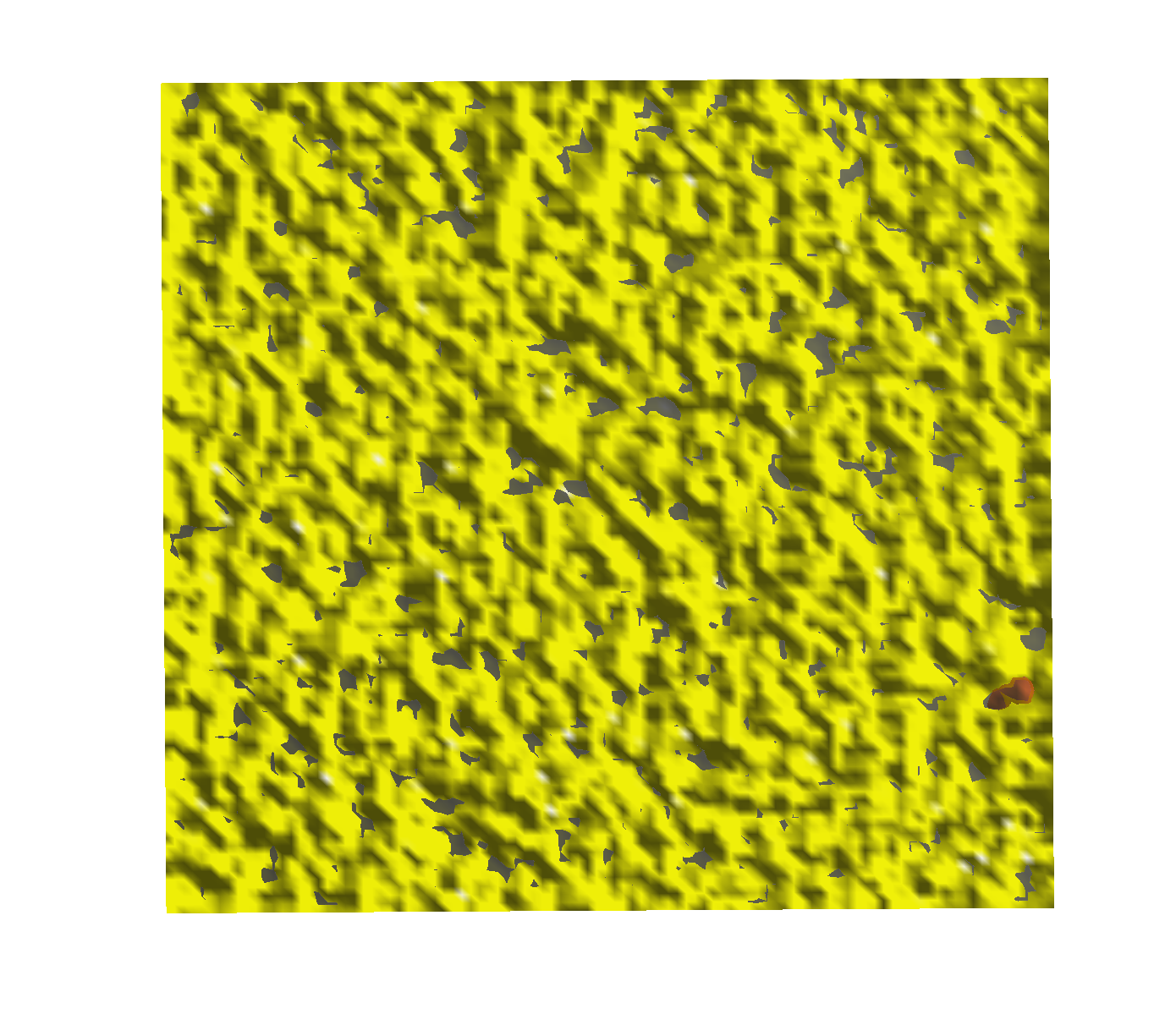}
	}
	\caption{The effective wells (in gray scale in (b)) $J_\delta(\mu_1)=\left \{ 1/u_n\le \mu_1+\delta \right\}$ with first eigenfunction superimposed in orange scale, where $\mu_1=1.3799$ and $\delta=0.05$.}
	\label{fig:2Dcon}
\end{figure}

In the Figure  \ref{fig:2Dcon} (a), 
it can be seen obviously that the first, localized eigenfunction $\phi^{(1)}$, which is plotted in orange scale, can be captured by one of the effective wells $J_{\delta}(\mu_1)$.
Figure \ref{fig:2Dcon} (b) also shows the contour of $J_{\delta}(\mu_1)$ with the first eigenfunction $\phi^{(1)}$.

\vskip 0.08 in


\Addresses


\begin{thebibliography}{999} \normalsize


\bibitem{Ab} E. Abrahams, ed., {\it 50 years of Anderson localization}, World Scientic Publishing Co. Pte. Ltd., Hackensack, NJ, 2010.

\bibitem{Ag}  Agmon, S., {\it Lectures on exponential decay of solutions of second-order elliptic equations: bounds on eigenfunctions of N-body 
Schr\"odinger 
operators.} Mathematical Notes, 29. Princeton University Press, Princeton, NJ; University of Tokyo Press, Tokyo, 1982.



\bibitem{An} P. W. Anderson, {\it Absence of diffusion in certain random lattices,} Physical Review, 109 (1958), pp. 1492--1505.


\bibitem{ADFJM-PRL} Douglas N. Arnold, Guy David, David Jerison, Svitlana Mayboroda, and Marcel Filoche. {\it Effective confining potential of quantum states in disordered media.} Physical Review Letters, 116(5), 2016.


\bibitem{ADFJM-CPDE} Douglas N. Arnold, Guy David, Marcel Filoche, David Jerison, and Svitlana Mayboroda. {\it Localization of eigenfunctions via an effective potential.} Communications in Partial Differential Equations, p. 1--31, 2019.

\bibitem{ADFJM-SIAM} Douglas N. Arnold, Guy David, Marcel Filoche, David Jerison, and Svitlana Mayboroda. {\it Computing spectra without solving eigenvalue problems.}  SIAM Journal of Scientific Computing, 41(1), B69–B92. (24 pages) 2019.

\bibitem{AM} M. Aizenman and S. Molchanov, {\it Localization at large disorder and at extreme energies: an elementary derivation,} Comm. Math. Phys., 157 (1993), pp. 245--278.

\bibitem{ASFH}
M. Aizenman, J. Schenker, R. Friedrich, D. Hundertmark, 
{\it Finite-volume fractional-moment criteria for
Anderson localization,}
 Commun. Math. Phys. 224 (2001), 219–253.


\bibitem{AW}
M. Aizenman, S.Warzel, {\it Random Operators: Disorder Effects on Quantum Spectra and
Dynamics.}  Graduate Studies in Mathematics, vol. 168. American Mathematical Society,
Providence (2015)

\bibitem{ANSS}
H. Abdul-Rahman, B. Nachtergaele, R. Sims, and G. Stolz,
{\it Localization properties of the disordered XY spin chain.}
Ann. Phys. (Berlin) 529, No. 7, 1600280 (2017)







\bibitem{BJ}
J. Bourgain and S. Jitomirskaya,
{\it Continuity of the Lyapunov exponent for quasiperiodic operators with analytic potential}
Journal of statistical physics 108.5-6 (2002): 1203-1218.

\bibitem{BK} J. Bourgain and C. Kenig,
{\it On localization in the continuous Anderson-Bernoulli model 
in higher dimension. } Invent. Math. 161 (2005), no. 2, 389--426. 



\bibitem{BLG} S. Balasubramanian, Y. Liao, and V. Galitski. {\it Many-Body Localization Landscape}. Phys. Rev. B 101, 014201, 2020.







\bibitem{Chung}F. R. K. Chung. {\it Spectral graph theory,} CBMS Regional Conference
series in mathematics, vol. 92, Washington, DC, 1994.

\bibitem{CKM}
R. Carmona, A. Klein, F. Martinelli, 
{\it Anderson localization for Bernoulli and other singular potentials,}
Commun. Math. Phys. 108 (1987), 41–66.


\bibitem{DSS}
D. Damanik, R. Sims, G. Stolz, 
{\it Localization for one-dimensional, continuum, Bernoulli-Anderson models,}
Duke Math. J. 114 (2002), 59-100.

\bibitem{DS}
J. Ding and C. Smart, {\it Localization near the edge for the Anderson Bernoulli model on the two dimensional lattice.}  Invent. Math. 219 (2020), no. 2, 467–506.


\bibitem{DFM}  Guy David, Marcel Filoche, Svitlana Mayboroda. {\it The landscape law for the integrated density of states.} arXiv:1909.10558



\bibitem{EA} Elgart, A.; Klein, A. {\it An eigensystem approach to Anderson localization.} J. Funct. Anal. 271 (2016), no. 12, 3465--3512. 






\bibitem{FM-PNAS}  Filoche, M.; Mayboroda, S., {\it Universal mechanism for Anderson and weak localization}. Proc. Natl. Acad. Sci. USA 109 (2012), no. 37, 14761--14766.



\bibitem{FK1} Figotin, A.; Klein, A.
{\it Localization of classical waves. I. Acoustic waves.} 
Comm. Math. Phys. 180 (1996), no. 2, 439--482. 

\bibitem{FK2} Figotin, A.; Klein, A.
{\it Localization of classical waves. II. Electromagnetic waves.}
Comm. Math. Phys. 184 (1997), no. 2, 411--441. 




\bibitem{FS} J. Frohlich and T. Spencer, {\it Absence of diffusion in the Anderson tight binding model
for large disorder or low energy,} Comm. Math. Phys., 88 (1983), pp. 151--184.





\bibitem{G} M. Gaudie. {\it Harmonic functions on square lattices: uniqueness sets and growth properties.} PhD thesis. Norwegian University of Science and Technology, Trondheim (2013) 




\bibitem{H} B. Helffer, {\it 
Semi-classical analysis for the Schr\"odinger  operator and applications.}
Lecture Notes in Mathematics, 1336. Springer-Verlag, Berlin, 1988.



\bibitem{KS}
 H. Kunz, B. Souillard, 
 {\it Sur le spectre des op\"erateurs aux diff\"erences finies al\"eatoires,} Commun. Math. Phys. 78 (1980), 201-246.




\bibitem{J}S. Jitomirskaya, {\it Metal-Insulator transition for the almost Mathieu operator.} Annals of Math. 150, 1159-1175 (1999).





\bibitem{LMF} Marcelo Lyra, Svitlana Mayboroda and Marcel  Filoche. {\it Dual hidden landscapes for Anderson localization in discrete lattices}, Europhys. Lett. EPL, Volume 109, Number 4, Editor's choice.



\bibitem{Kir} W.\, Kirsch, 
{\it An invitation to random Schr\"odinger operators. }
With an appendix by Fr\'ed\'eric 
Klopp. Panor. Synth\`eses, 25, Random Schr\"odinger operators, 1--119, Soc. Math. France, Paris, 2008. 




\bibitem{PRL-3} 
Lefebvre, G.; Gondel, A.; Dubois, M.; Atlan, M.; Feppon. F.; Labb\'e, A.; Gillot, C.; Garelli, A.; Ernoult, M.; Mayboroda, S.; Filoche, M.; Sebbah, P. {\it One single static measurement predicts wave localization in complex structures.} 
Phys. Rev. Lett., to appear.


\bibitem{LS} Lu, J.; Steinerberger, S. {\it  
Detecting localized eigenstates of linear operators. }
Res. Math. Sci. 5 (2018), no. 3, 5:33. 











\bibitem{St} Steinerberger, S. {\it  Localization of quantum states and landscape functions.} Proc. Amer. Math. Soc. 145 (2017), no. 7, 2895--2907. 

\end{thebibliography}
\end{document}